\documentclass[12pt, draftclsnofoot,onecolumn]{IEEEtran}
\usepackage{graphicx}
\usepackage{cite}
\usepackage{enumerate}
\usepackage{array}
\usepackage{amsmath}
\usepackage{amssymb}
\usepackage{subeqnarray}
\usepackage{cases}
\usepackage{bm}
\usepackage{algorithm}
\usepackage{algorithmic}
\usepackage{multirow}
\usepackage{epsfig}
\usepackage{epstopdf}
\usepackage{xcolor}
\usepackage{subfigure}

\newtheorem{theorem}{Theorem}
\newtheorem{lemma}{Lemma}

\newtheorem{corollary}{Corollary}

\newenvironment{proof}{{\bf Proof:}}{\hfill\rule{2mm}{2mm}}
\newcommand{\Prob}[1]{\mathrm{P}\left\{#1\right\}}

\begin{document}
\title{Closed-Form Analysis of Non-Linear Age-of-Information in Status Updates with an Energy Harvesting Transmitter}

\author{Xi~Zheng, Sheng~Zhou,~\IEEEmembership{Member,~IEEE,} Zhiyuan~Jiang,~\IEEEmembership{Member,~IEEE,} Zhisheng~Niu,~\IEEEmembership{Fellow,~IEEE}
\thanks{This work is sponsored in part by the Nature Science Foundation of China (No. 61571265, No. 91638204, No. 61701275, No. 61861136003, No. 61621091), and Hitachi Ltd. (\emph{Corresponding author: Sheng Zhou.})

The authors are with Beijing National Research Center for Information Science and Technology, Department of Electronic Engineering, Tsinghua University, Beijing 100084, China. Emails: zhengx14@mails.tsinghua.edu.cn,
\{zhiyuan, sheng.zhou, niuzhs\}@tsinghua.edu.cn.}}

\maketitle

\begin{abstract}
Timely status updates are crucial to enabling applications in massive Internet of Things (IoT). This paper measures the data-freshness performance of a status update system with an energy harvesting transmitter, considering the randomness in information generation, transmission and energy harvesting. The performance is evaluated by a non-linear function of age of information (AoI) that is defined as the time elapsed since the generation of the most up-to-date status information at the receiver. The system is formulated as two queues with status packet generation and energy arrivals both assumed to be Poisson processes. With negligible service time, both First-Come-First-Served (FCFS) and Last-Come-First-Served (LCFS) disciplines for arbitrary buffer and battery capacities are considered, and a method for calculating the average penalty with non-linear penalty functions is proposed. The average AoI, the average penalty under exponential penalty function, and AoI's threshold violation probability are obtained in closed form. When the service time is assumed to follow exponential distribution, matrix geometric method is used to obtain the average peak AoI. The results illustrate that under the FCFS discipline, the status update frequency needs to be carefully chosen according to the service rate and energy arrival rate in order to minimize the average penalty.
\end{abstract}


\section{Introduction}
Satisfying strict real-time requirements in wireless communication systems is of extensive concerns. A typical real-time application in Internet of Things (IoT) is remote monitoring and control, which requires timely status update to the fusion center, i.e., status information of the objects should be refreshed at the fusion center in a timely manner. However, due to the inevitable delays in queuing and transmission, the received status packets do not carry the present status information.
To characterize the lag in status update, \textit{age of information} (AoI) has been proposed in \cite{6195689} as a metric for information freshness. It is defined as the time that has elapsed since the generation of the most up-to-date status information at the receiver, as is illustrated in Fig. \ref{fig:age}. The AoI at time $t$ is expressed as $$\Delta(t) = t - U(t), $$ where $U(t)$ represents the time stamp at the generation epoch of the {\em most} up-to-date status information that has been received before time $t$. In \cite{6195689}, the status update process is formulated as a queuing system, in which the traffic arrivals correspond to the generations of status packets, and the service times correspond to transmissions and the time waiting for medium access. The time-averaged AoI of M/M/1, M/D/1 and D/M/1 queues under the first-come-first-served (FCFS) discipline is obtained, and the comparison among the three queues indicates that a more regular update brings a smaller average AoI.

\begin{figure}[h]
\centering
\vspace*{-0.2in}
\includegraphics[width=2.8in]{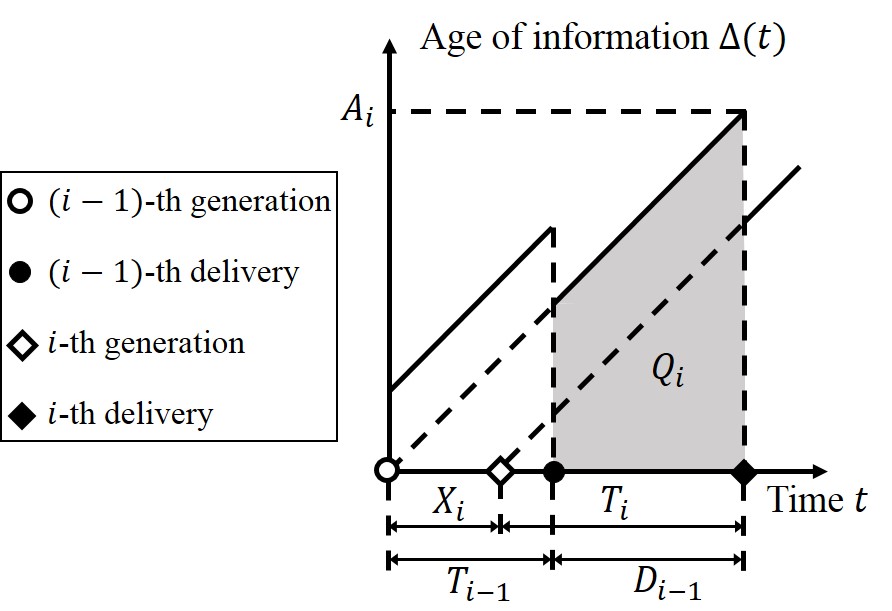}
\vspace*{-0.2in}
\caption{An example that illustrates the evolution of AoI. }
\label{fig:age}
\vspace*{-0.1in}
\end{figure}

Meanwhile, numerous nodes in IoT will be powered by renewable sources using energy harvesting techniques \cite{5522465}. For example, to sustainably update the environmental conditions (temperature, air humidity and etc.), motion (position and velocity) or other characters to a fusion center, energy harvesting transmitters can free the monitored object from limited power backup or cables. Wireless nodes powered by renewable energy are faced with the randomness in energy supplies, which makes guaranteeing real-time communications more challenging compared with grid-powered or solely battery-powered systems, and further complicates the problem of keeping status update in a regular fashion. Therefore, the performance analysis of a status update system with an energy harvesting transmitter is essential in addition to the analysis of a system powered by a stable energy source.

The average AoI is considered a key indicator to the performance of a status update system. However, it cannot straightly describe the performance degradation caused by the lag in status update. A thorough survey on the performance degradation caused by information staleness can be found in \cite{NonlinearFunc}. The performance degradation of a system, such as the inaccuracy in monitoring and the invalidity of control decisions, varies based on the scenario and is generally non-linear on AoI. Ref. \cite{DBLP:journals/corr/abs-1805-03271} shows that the characteristics of the average AoI and the violation probability of {\em peak AoI}, which is defined as the AoI right before the reception of a new status information, can be noticeably different. Therefore, the actual performance of a status update system, which is non-linear on AoI, should be further investigated. The non-linear AoI penalty function is introduced in \cite{8000687} to evaluate information staleness. Ref. \cite{8006543} extends the work on M/M/1 in \cite{6195689}, and obtains the average penalty under exponential and logarithmic AoI-penalty functions. The case where AoI is expected to be within a certain threshold is considered in \cite{8424039}, and scheduling algorithms to allocate channel access among multiple users are proposed in order to reduce violation probability.
Other work mainly focuses on the analysis on average AoI or average peak AoI\cite{7364263}--\!\!\cite{DBLP:journals/corr/abs-1801-03975}. Two special cases of multiple servers, M/M/$\infty$ and M/M/2 systems, are investigated in \cite{7364263}. Time-averaged peak AoI is considered for multi-class M/G/1 and M/G/1/1 queuing systems in \cite{7282742} and M/M/1 queues with delivery errors in \cite{7541765}. Ref. \cite{DBLP:journals/corr/abs-1805-12586} studies the average AoI of a G/G/1/1 status update system. Ref. \cite{7541763} studies the packet management policies for a status update system and shows that preemptive LCFS discipline is age-optimal, throughput-optimal and delay-optimal, given that the service times are independent and identically distributed (i.i.d.) exponential random variables. Ref. \cite{7415972} compares the performance of M/M/1, M/M/1/1 and M/M/1/2*, in which the backlogged packet is replaced if a new packet arrives. The stationary distribution of AoI in GI/GI/1 system is derived in \cite{SamplePath}. Recently, Ref. \cite{DBLP:journals/corr/YatesK16} introduces a method named stochastic hybrid systems (SHS) to the analysis of status updating systems. For M/M/1 queues with multiple sources, the performance under FCFS, preemptive and non-preemptive last-come-first-served (LCFS) policy is investigated with SHS method.
Ref.\cite{DBLP:journals/corr/abs-1801-01803}--\!\!\cite{DBLP:journals/corr/abs-1801-03975} investigate user-scheduling policies to minimize overall average AoI in multi-user scenarios.

Scheduling schemes to reduce latency in an energy harvesting communication system can be found in \cite{7010878}--\!\!\cite{7492927}. Efforts on the AoI in energy harvesting powered status update systems began to emerge recently. Existing work mainly focuses on the scheduling of status packet transmissions subject to energy constraints \cite{7308962}--\!\!\cite{Arafa18}, where the generation of status packets can be fully controlled. Ref. \cite{7308962} considers scheduling policies for the minimization of average AoI under energy replenishment constraints in the systems where the queuing delay and service time are neglected. Ref. \cite{7283009} proposes and compares three intuitive status update policies, which try to equalize update interval, to equalize update delay, and to reduce packet queuing, respectively. Ref. \cite{Yang} analyzes energy harvesting systems with different battery capacities, and discusses their AoI-optimal transmission scheduling policies by which the average AoI is minimized. A two-hop status update system with energy harvesting transmitter and relay is investigated in \cite{DBLP:journals/corr/ArafaU17}. Ref. \cite{8006703} analyzes the battery-threshold policy in energy harvesting system, and finds the condition for one to minimize the average AoI among all threshold policies. The optimality of energy dependent AoI-threshold policies in a system with finite battery and zero service time is proved in \cite{Arafa18}. Ref. \cite{8406846} and \cite{8437904} explore M/M/1/1 energy harvesting status update systems, where there are at most one status packet buffered in the system. They investigate the average AoI in the system with SHS method, and analyze the asymptotic cases where status packet arrival rate, energy arrival rate, or service rate is relatively large, respectively.

In this paper, we jointly consider the randomness in status packet generation and energy harvesting in the analysis of non-linear AoI-based performance for an energy harvesting wireless communication system. The randomness in transmission and MAC delay is also considered to investigate the overall impact of status packet generation, energy arrivals and services on the AoI of a status update system. The status update problem is characterized as a queuing system with a finite data queue (buffer), where the status packets are stored after their generations, and a finite energy queue (battery). Both the status packet arrivals and energy packet arrivals are assumed to be Poisson processes. The problem considered is similar to the one in \cite{8406846} and \cite{8437904}, and they focus solely on average AoI in a system with unit buffer size, and investigate the average AoI by SHS. Different from \cite{8406846} and \cite{8437904}, this paper applies the conventional stochastic analysis with queuing model, and analyzes the non-linear AoI-based performance arbitrary buffer and battery capacity. The contribution of this paper is summarized as follows:
\begin{enumerate}
\item A method to obtain the closed-form AoI-based penalty is established for both FCFS and LCFS disciplines when the service time is negligible, since the service time is usually much smaller than packet generation intervals and energy arrival intervals in real-world applications. Explicitly, the closed-form expressions of the cumulative probability distributions (CDF) of the peak AoI and the sojourn time, and the rate of valid updates\footnote{A valid update is defined as a status packet that is the most up-to-date status packet upon reception in Sec. \ref{sec:model}. } are obtained. The average non-linear penalty of the system can be derived by this method for integrable AoI-penalty functions.
\item Average penalties under three typical AoI-penalty functions are obtained and analyzed. They corresponds to the average AoI, average exponential penalty of AoI and AoI's threshold violation probability in the system. The results are further compared under different buffer capacities, battery capacities and service disciplines. Results shows that the exponential penalty is extremely sensitive to the ratio between status generation frequency and energy arrival rate, especially under the FCFS discipline when the buffer capacity is large.
\item To consider a more general case and incorporate non-negligible service time, the service time is assumed to be independent and identically distributed (i.i.d.) random variables following exponential distribution, and the status update is formulated as a quasi-birth-and-death (QBD) process with finite battery capacity. The explicit expression for the average peak AoI is obtained after the stationary distribution of system states is computed by the matrix geometric method \cite{neuts1994matrix}.
\end{enumerate}

The remainder of the paper is organized as follows. Section \ref{sec:model} explains the basic notations and the problem formulation.  Results in the negligible-service-time regime are described in Section \ref{sec:a}. Section \ref{sec:non-zero} formulates the problem as a QBD process and describes how to compute the average peak AoI. Section \ref{sec:anal} illustrates the results with figures. Section \ref{sec:con} concludes the paper.

\section{System Model}
\label{sec:model}

As depicted in  Fig. \ref{fig:queue}, a status update system powered by renewable energy is modeled as a queuing system with a data buffer and a battery. The arrival of a status packet corresponds to the generation of a status information from the source, while each departure represents the successful reception of the status information at the receiver side.

\begin{figure}[t]
\centering
\includegraphics[width=2.5in]{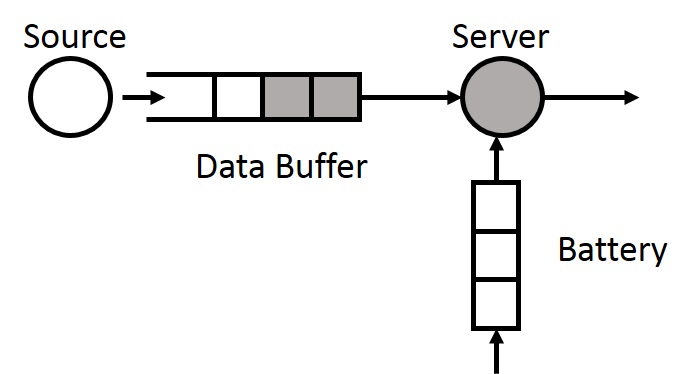}
\vspace*{-0.1in}
\caption{Queuing model for the analysis of a status update system with an energy harvesting server.}
\label{fig:queue}
\vspace*{-0.3in}
\end{figure}

\subsection{Status Update Model}
To capture the randomness in the generation of status packets, the arrival of status packets is assumed to be a Poisson process with rate $\lambda$. After its generation, a status packet first arrives at the data buffer, which can store at most $K$ status packets. Thus, if there are already $K$ status packets in the data buffer upon the arrival of a new status packet, one status packet must be dropped.

The server is assumed to be work-conserving, \textit{i.e.}, the server is idle only when at least one of the two queues is empty. Two kinds of service disciplines are considered. The first one is FCFS, with which a newly-arrived status packet waits at the end of the data queue, and will not be served until all the status packets ahead of it depart. Under the FCFS discipline, older status packets have higher priority, and the newly-arrived status packet will be blocked out of the data buffer if there are already $K$ status packets in the data buffer. The other service discipline is LCFS, which allows the latest status packet to wait at the front of the queue and the oldest packet in the buffer to be discarded if the buffer is full.

It is possible that a newly-delivered status packet is older than the status packets at the receiver side. However, the reception of old information does not change the AoI. To clarify this, a \emph{valid update} is defined as a status packet that is delivered and is the most up-to-date packet upon its reception. Thus, the AoI is reset to the age of the received packet only when it is a valid update. In \cite{SamplePath}, the authors define informative packets and non-informative packets to distinguish the status packets that reset the AoI from those that do not. The informative packets in \cite{SamplePath} corresponds to the valid updates. Under FCFS discipline, any status packet that enters the data buffer is a valid update, while under LCFS discipline, a status packet is a valid update if and only if there is no status packet arrival between its arrival and departure.

\subsection{Energy Model}
The energy model in this paper is similar to those in \cite{7283009}--\!\!\cite{8437904}. The transmitter is powered by an energy harvesting module, which consistently harvests energy and stores it in a battery with limited capacity. Assume that the service of a status packet requires $E_0$ Joules of energy, which is also referred to as an energy packet. Discretizing the battery by $E_0$ Joules of energy, the capacity of the battery is denoted as $B$ energy packets. Therefore, the arrival of energy packets represents the accumulation of integral multiples of $E_0$ Joules in the battery. The arrival of energy packets is modeled as a Poisson process with rate $r$, which characterizes the randomness and unpredictability in energy harvesting. The energy arrival rate $r$ is assumed to be greater than the status packet arrival rate $\lambda$ to ensure the stability of the data queue. A new energy packet is immediately discarded if the battery is full.

\subsection{Objective Functions}
\subsubsection{Linear functions}
\label{sss:linear}
The performance of a status update system is commonly evaluated by the long-term averaged AoI or the long-term averaged peak AoI. Denote the inter-arrival time between the $(i-1)$-th and the $i$-th valid update by $X_i$, the inter-departure time between the $i$-th and the $(i+1)$-th valid update by $D_i$ and the $i$-th valid update's sojourn time by $T_i$, as is depicted in Fig. \ref{fig:age}.
The $i$-th \textit{peak AoI} $A_i$, which is the AoI right before the reception of the $i$-th valid update \cite{7415972}, can be written as
\begin{equation}
\label{Eq:pAoI}
A_i = X_i + T_i = D_{i-1}+T_{i-1}.
\end{equation}
Eq. (\ref{Eq:pAoI}) can be proved straightforward by the evolution curve of AoI in Fig. \ref{fig:age}.
Taking expectation, letting $i\to\infty$ at both sides of Eq. (\ref{Eq:pAoI}) and denoting the limits of the variables as the ones without subscripts, the average peak AoI follows
\begin{eqnarray}
\label{Eq:page}
\mathbb{E}\left[A\right] = \mathbb{E}\left[X\right] + \mathbb{E}\left[T\right] = \mathbb{E}\left[D\right] + \mathbb{E}\left[T\right].
\end{eqnarray}
According to \cite{6195689}, the average AoI of a stationary and ergodic system is defined as $$\bar{\Delta} = \lim_{\mathcal{T}\to\infty}\frac{1}{\mathcal{T}}\int_{0}^{\mathcal{T}}\Delta(t)\,\mathrm{d}t, $$in which the integral equals the area below the AoI curve. Denote the area below the AoI curve between the delivery of the $(i-1)$-th and the $i$-th valid update as $Q_i$, as is shadowed in Fig. \ref{fig:age}. The average AoI is expressed as
$$\bar{\Delta} = \lim_{\mathcal{T}\to\infty}\frac{1}{\mathcal{T}}\sum_{i = 1}^{N(\mathcal{T})} Q_i, $$
where $N(\mathcal{T})$ represents the number of valid updates delivered before $t=\mathcal{T}$. Defining the arrival rate of valid updates as $\tilde{\lambda} = \lim_{\mathcal{T}\to\infty}\frac{N(\mathcal{T})}{\mathcal{T}}$, the average AoI becomes $\bar{\Delta} = \tilde{\lambda}\lim_{\mathcal{T}\to\infty}\mathbb{E}\left[Q_i\right]. $
Since $Q_i = \left(A_i^2-T_{i-1}^2\right)/2$, the average AoI is given by
\begin{eqnarray}
\bar{\Delta} &=& \frac{\tilde{\lambda}}{2}\left(\mathbb{E}\left[A^2\right] - \mathbb{E}\left[T^2\right]\right).
\label{Eq:ageo}
\end{eqnarray}

\subsubsection{Non-linear functions}
Considering the non-linear performance degradation caused by outdated data, AoI is further generalized as an AoI-related penalty function, which quantizes the performance based on AoI. A general way to characterize the non-linear penalty is by defining a non-linear function $g(\Delta)$ that maps AoI to the penalty as introduced in \cite{8000687}, so that the average penalty of a stationary and ergodic system is
\begin{eqnarray}
\label{Eq:general}
C &=& \lim_{\mathcal{T}\to\infty}\frac{1}{\mathcal{T}}\int_{0}^{\mathcal{T}}g\left(\Delta\left(t\right)\right)\,\mathrm{d}t.
\end{eqnarray}
In \cite[Example 1]{SamplePath}, the authors provide how to obtain the average penalty with the CDF of AoI given, without further analysis on the average penalty. It is also proved that the CDF of AoI can be obtain with the CDFs of the peak AoI $A$ and the sojourn time $T$, and the arrival rate $\tilde{\lambda}$ of valid updates. Next, we are going to provide another method to show how to derive the average penalty directly with the CDFs of the peak AoI $A$ and the sojourn time $T$, and the arrival rate $\tilde{\lambda}$ of valid updates.

The right-hand side of Eq. (\ref{Eq:general}) is rearranged as the summation of integrals over the intervals between sequential deliveries of valid updates:
\begin{eqnarray*}
\lim_{\mathcal{T}\to\infty}\frac{1}{\mathcal{T}}\int_{0}^{\mathcal{T}}g\left(\Delta\left(t\right)\right)\,\mathrm{d}t &=& \lim_{\mathcal{T}\to\infty}\frac{N(\mathcal{T})}{\mathcal{T}}\lim_{\mathcal{T}\to\infty}\frac{1}{N(\mathcal{T})}\sum_{i = 0}^{N(\mathcal{T})-1}\int_{\sum_{j=0}^{i-1}D_j}^{\sum_{j=0}^{i}D_j}g\left(\Delta\left(t\right)\right)\,\mathrm{d}t.
\end{eqnarray*}
According to the definition, AoI is set to sojourn time $T_i$ upon the $i$-th valid update's delivery, and grows linearly with unit slope before the delivery of the $(i+1)$-th valid update. Therefore,
\begin{eqnarray*}
\lim_{\mathcal{T}\to\infty}\frac{1}{\mathcal{T}}\int_{0}^{\mathcal{T}}g\left(\Delta\left(t\right)\right)\,\mathrm{d}t&=& \tilde{\lambda}\lim_{\mathcal{T}\to\infty}\frac{1}{N(\mathcal{T})}\sum_{i = 0}^{N(\mathcal{T})-1}\int_{T_i}^{A_i}g\left(\Delta\right)\,\mathrm{d}\Delta \\
&=& \tilde{\lambda}\left(\mathbb{E}\left[G(A)\right] - \mathbb{E}\left[G(T)\right]\right),
\end{eqnarray*}
in which $G(x) = \int_0^xg(\Delta) \,\mathrm{d}\Delta$.

Thus, given the CDFs of the peak AoI $A$ and the sojourn time $T$, and the arrival rate $\tilde{\lambda}$ of valid updates, the average penalty of the system can be determined by
\begin{equation}
\label{Eq:int}
C = \tilde{\lambda}\int_0^{\infty}G(a)\,\mathrm{d}\mathrm{P}\left\{A\leq a\right\} - \tilde{\lambda}\int_0^{\infty}G(t)\,\mathrm{d}\mathrm{P}\left\{T\leq t\right\}.
\end{equation}

The average penalty indicates the long-term average performance of a status update system. Our objective is to investigate how system parameters, such as status packet arrival rate, buffer capacity, battery capacity and service disciplines, affect the average penalty.

\section{Negligible Service Time Regime}
\label{sec:a}
To gain more insights, in this section, we first explore the asymptotic results where the service time is negligible compared to the status packet arrival intervals and the energy packet arrival intervals. The scenario where the average service time is much shorter than status packet arrival intervals and energy arrival intervals commonly exists in practical IoT applications. The service time mainly incorporates transmission and the MAC (Medium Access Control) delay. Since status update packets are usually small in data size, the time for transmission (even if retransmission is considered) or MAC delay is relatively short. For example, a sub-frame in FDD-LTE is 1ms, so the service time of a short packet can be several milliseconds. On the contrary, the need for status update (e.g., temperature) and energy harvesting are mostly in a larger time scale than milliseconds. Therefore, the asymptotic regime where the service time is negligible is reasonable to consider and significant in offering insight on the performance of status update, especially on the impact of energy provision in such a system.

\begin{figure}[h]
\centering
\includegraphics[width=3.5in]{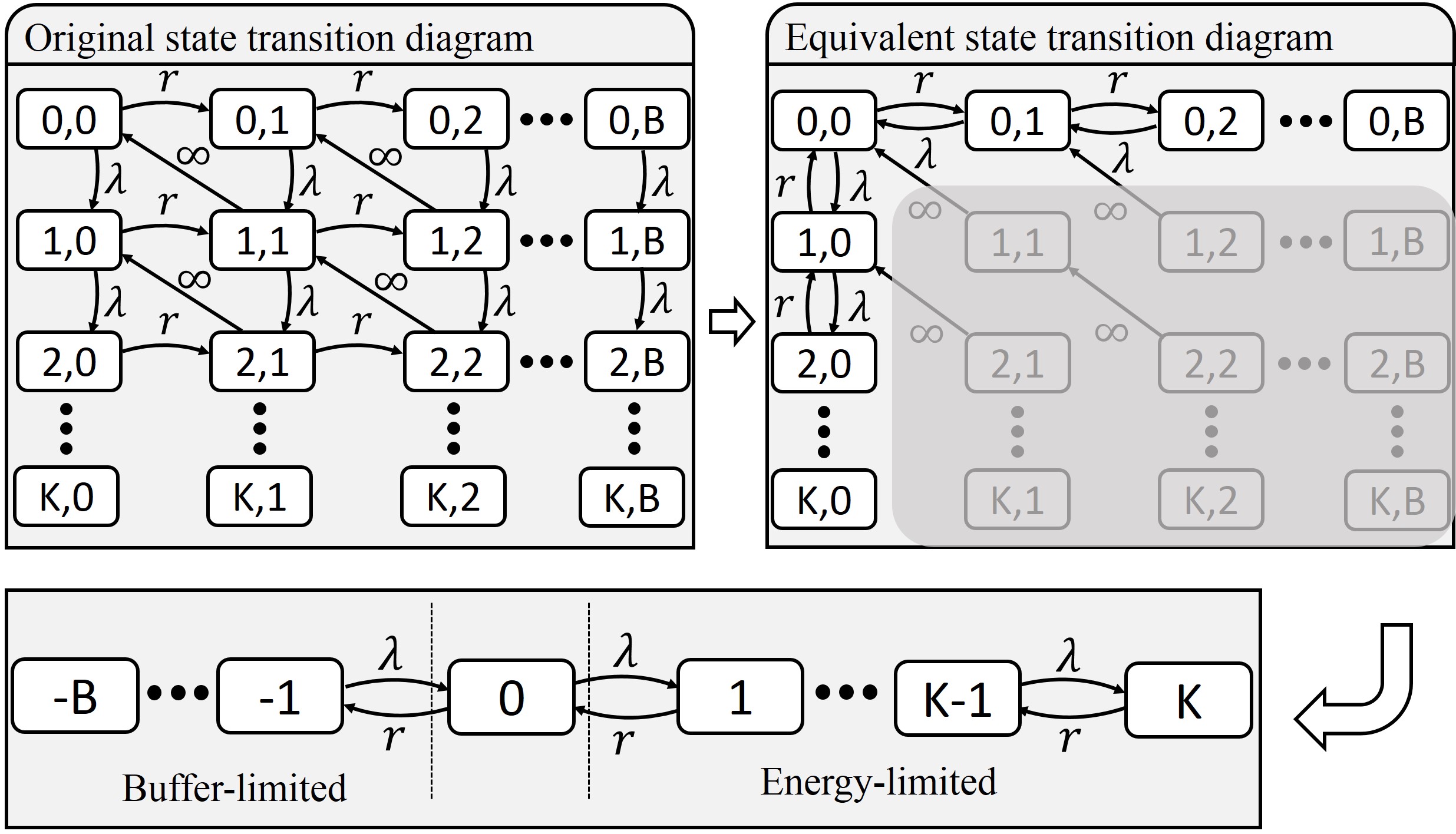}
\caption{The state transition diagram for the negligible-service-time problem. }
\label{fig:one-dim-transit}
\vspace*{-0.1in}
\end{figure}

Denote the data queue length by $q_1$ and the energy queue length by $q_2$. The state transition in the negligible-service-time regime is illustrated in the left diagram in Fig. \ref{fig:one-dim-transit}. If data queue length $q_1 < K$, each state $(q_1, q_2)$ transits to state $(q_1+1, q_2)$ with rate $\lambda$; If energy queue length $q_2 < B$, each state $(q_1, q_2)$ transits to state $(q_1, q_2+1)$ with rate $r$. Note that when the service time is negligible, the service for a status packet can be completed instantly when the energy queue is not empty. Therefore, every state $(q_1,q_2)$ with both $q_1>0$ and $q_2>0$ transits to state $(q_1-1,q_2-1)$ with rate $\infty$. In other words, any state $(q_1,q_2)$ with $q_1q_2\neq0$ is a transient state that will never be visited thereafter. In the right diagram of Fig. \ref{fig:one-dim-transit}, an equivalent state transition is drawn with all the transient states shadowed in gray. As shown in the diagram, the state transition among the recurrent states, \textit{i.e.}, $q_1q_2=0$, is identical to an M/M/1 queue. Based on this, the system states are indexed as $S = q_1 - q_2 \in \left\{-B, -\left(B-1\right), \cdots,-1,0,1,\cdots, K\right\}$, which distinctively maps a recurrent state to $S$:
\begin{enumerate}
\item If $S < 0$, data queue length $q_1 = 0$ and energy queue length $q_2 = -S$;
\item If $S = 0$, both queues are empty;
\item If $S > 0$, data queue length $q_1 = S$ and energy queue length $q_2 = 0$.
\end{enumerate}

For each state $S \in \left\{-(B-1),\cdots, -1,0,1\cdots, K\right\}$, the system transits to state $S-1$ when there is an energy arrival; for $S \in \left\{-B, \cdots,-1,0,1,\cdots, K-1\right\}$, system state becomes $S+1$ when there is a status packet arrival. Thus, as is shown in the bottom diagram of Fig. \ref{fig:one-dim-transit}, the system is equivalent to an M/M/1 queue with status packet arrival rate $\lambda$, service rate $r$ and buffer size $K+B$. Let $\theta$ denote the ratio of the status packet arrival rate to the energy packet arrival rate, \textit{i.e.}, $\theta = \frac{\lambda}{r}$, and the utilization of the M/M/1 queue is $\theta$. Next, we first derive the stationary probability distribution of the peak AoI $A$ and the sojourn time $T$, and the arrival rate $\tilde{\lambda}$ of valid updates, then obtain the average penalty under three types of penalty functions (depicted in Fig. \ref{fig:func}) by Eq. (\ref{Eq:int}). The three penalty functions are:
\begin{enumerate}
\item Linear function $g(\Delta) = \Delta$: The average penalty equals the long-term average AoI. This penalty function is suitable for the systems in which the influence of the delay in information grows approximately linearly with time.
\item Exponential function $g(\Delta) = \alpha^{-1}\left(e^{\alpha\Delta}-1\right), \alpha\neq0$: We let $g(0) = 0$ and $g'(0) = 1$, so that the penalty and the slope when $\Delta=0$ are identical to the ones in the linear function, for the simplicity of comparison. When exponent $\alpha>0$, the exponential penalty function is more sensitive to large AoI compared to the linear function, and is more reasonable in the scenarios where the time-correlation among status is relatively small. Taking the finite-state Markov channel\cite{350282}, which is a widely adopted channel model, as an example, the probability of correctly estimating the current channel state based on delayed channel state exponentially decays to the limiting distribution. Additionally, note that the exponential penalty function is increasing for both $\alpha>0$ and $\alpha<0$. Exponent $\alpha$ being negative can be useful in systems such as \cite{NegExpo}.
\item Shifted unit step function $g(\Delta) = \mathbf{1}\left\{\Delta-\beta\right\}$: Unit step function $\mathbf{1}\left\{\cdot\right\}$ is defined as
\begin{eqnarray*}
\mathbf{1}\left\{x\right\} = \left\{
\begin{aligned}
&1, \mathrm{~if~}x>0,\\
&0, \mathrm{~if~}x \leq 0.
\end{aligned}
\right.
\end{eqnarray*}
Under the shifted unit step function, the long-term average penalty gives the probability on AoI exceeding threshold $\beta$. This function should be considered when there is a certain upper-bound constraint of AoI that the system tries not to violate.
\end{enumerate}
\begin{figure}[t]
\centering
\includegraphics[width=3in]{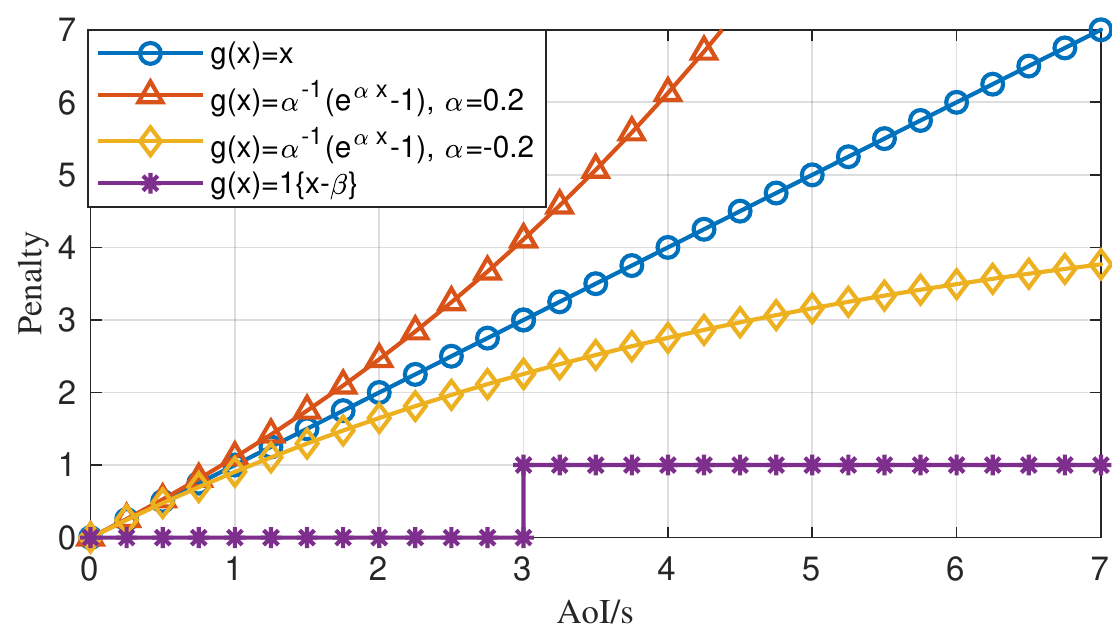}
\vspace*{-0.15in}
\caption{Three penalty functions to be analyzed. In the figure, the parameters are $\alpha = \pm0.2$ and $\beta = 3$. }
\label{fig:func}
\vspace*{-0.2in}
\end{figure}

In the following two subsections, the formulas for the average penalty under the three types of penalty functions will be listed with preliminary intuition towards the expressions. The results will be thoroughly analyzed together with figures in Section \ref{sec:anal}.
\subsection{FCFS Discipline}
\label{sec:F}
For FCFS descipline, we have the following results.
\begin{lemma}
\label{lemma:FCFS1}
Under the FCFS discipline, with status packet arrivals and energy packet arrivals being Poisson processes of rate $\lambda=r\theta$ and $r$, and the service time being negligible, the arrival rate $\tilde{\lambda}$ of valid updates is
$$\tilde{\lambda}_\mathrm{FCFS} = \lambda \frac{1-\theta^{K+B}}{1-\theta^{K+B+1}}. $$
\end{lemma}
\begin{proof}
See Appendix \ref{proof:l1}.
\end{proof}

The proof for Lemma \ref{lemma:FCFS1} is straightforward: Since all the status packets admitted to the data buffer are valid updates under FCFS, the arrival rate of valid update equals the arrival rate of status packets that enter the buffer. Therefore, the arrival rate $\tilde{\lambda}$ of valid update under the FCFS discipline equals the throughput.

\begin{lemma}
\label{lemma:FCFS}
Under the FCFS discipline, with status packet arrivals and energy packet arrivals being Poisson processes of rate $\lambda=r\theta$ and $r$, and the service time being negligible, the CDFs of the peak AoI and the sojourn time of valid updates are
\begin{eqnarray}
\label{Eq:AF}
\mathrm{P}\left\{A\leq a\right\} &=&  1 - \frac{e^{-\lambda a}\theta^{-B}}{\theta^{-B}-\theta^{K}} - \frac{\theta^{-1}e^{-ra}}{\theta^{-B}-\theta^{K}}\sum_{n=1}^{K}\frac{\left(\lambda a\right)^{n}}{n!}+ \frac{\theta^{K}e^{-ra}}{\theta^{-B}-\theta^{K}}\sum_{n=0}^{K}\frac{\left(ra\right)^{n}}{n!},
\end{eqnarray}
and
\begin{eqnarray}
\label{Eq:TF}
\mathrm{P}\left\{T\leq t\right\} &=& 1 - \frac{1}{\theta^{-B} - \theta^{K}}\sum_{n=0}^{K}\frac{\left(\lambda t\right)^n}{n!}e^{-rt} + \frac{\theta^{K}}{\theta^{-B}-\theta^{K}}\sum_{n=0}^{K}\frac{\left(rt\right)^n}{n!}e^{-rt},
\end{eqnarray}
respectively.
\end{lemma}
\begin{proof}
See Appendix \ref{proof:l2}.
\end{proof}

According to Eq. (\ref{Eq:int}), the average penalty under FCFS discipline can be obtained by Lemma \ref{lemma:FCFS1} and \ref{lemma:FCFS}. For the three special cases of penalty function being $g(\Delta) = \Delta$, $g(\Delta) = \alpha^{-1}\left(e^{\alpha\Delta}-1\right)$, and $g(\Delta) = \mathbf{1}\left\{\Delta-\beta\right\}$, we obtain the long-term average penalty by Eq. (\ref{Eq:int}) and summarize the results in Theorem \ref{thm:FL}--\ref{thm:FS}.
\begin{theorem}
\label{thm:FL}
Under the FCFS discipline, with status packet arrivals and energy packet arrivals being Poisson processes of rate $\lambda=r\theta$ and $r$, and the service time being negligible, the average AoI is
\begin{equation}
\label{Eq:aveF}
C_\mathrm{FCFS,~linear} = \lambda^{-1} + r^{-1}\frac{\theta}{\theta^{-B}-\theta^{K+1}}\left[-K\theta^{K} + \frac{1 + \theta^{K-1} - 3\theta^{K} + \theta^{K+1}}{1-\theta}\right].
\end{equation}
\end{theorem}

Consider an average-power-constrained scenario, where the average energy usage rate is limited by the energy arrival rate $r$. Since $\theta$ is less than $1$, each status packet can be served as soon as it arrives at the system, thus the average AoI is $\lambda^{-1}$. It is exactly the first term at the right-hand sides of Eq. (\ref{Eq:aveF}). As battery capacity $B$ increases, the average AoI decreases and approaches the average AoI under the average-power-constrained scenario. This result offers guidance to the selection of the battery capacity in an energy harvesting status update system. For example, given the system requirement that the average AoI should not exceed $\Delta_\mathrm{max}$, the battery capacity must be greater than
$$B_\mathrm{min} = \log_\theta\frac{\lambda\Delta_\mathrm{max}-1}{\lambda\theta^K\Delta_\mathrm{max} + \left(-K\theta^{K} - \theta^{K} + \frac{1 - \theta^{K}}{1-\theta}\right)\theta^2}. $$
Similar characteristics can be found in Theorems \ref{thm:FE}--\ref{thm:LS} for the other penalty functions as well as under the LFCS discipline.

Additionally, the average AoI under the FCFS discipline is derived in \cite{8406846}\cite{8437904} by SHS for a system where only one status packet is allowed in the system if the battery is not empty (identical to the case $K=0$ if service rate $\mu\to\infty$). The average AoI (see \cite[Eq. (7)]{8406846} and \cite[Eq. (9)]{8437904}) matches Eq. (\ref{Eq:aveF}) under the conditions $K=0$ and $\mu\to\infty$, which further verifies Theorem \ref{thm:FL}.

\begin{theorem}
\label{thm:FE}
Under the FCFS discipline, with status packet arrivals and energy packet arrivals being Poisson processes of rate $\lambda=r\theta$ and $r$, and the service time being negligible, with penalty function $g\left(\Delta\right) = \alpha^{-1}\left(e^{\alpha\Delta}-1\right), \alpha < \lambda$, the average penalty is
\begin{eqnarray}
\label{Eq:thmFE}
C_\mathrm{FCFS,~exp} =\left\{
\begin{aligned}
&\frac{1}{\lambda-\alpha} + \frac{r\alpha^{-1}}{\theta^{-B}-\theta^{K+1}} \left[\frac{\theta^{K+2}}{\lambda-\alpha} + \left(1-\theta\right)\frac{1 - \frac{\lambda^{K+1}}{\left(r-\alpha\right)^{K+1}}}{r-\alpha-\lambda} - \frac{1}{r-\alpha}\right], \mathrm{if}~\alpha \neq r-\lambda,\\
&\frac{1}{\lambda-\alpha} + \frac{r^{-1}}{\theta^{-B}-\theta^{K+1}} \left[\frac{\theta^{K+2}-2\theta+1}{(2\theta-1)(1-\theta)} + \frac{K}{\theta}\right], \mathrm{if}~\alpha = r-\lambda.
\end{aligned}
\right.
\end{eqnarray}
\end{theorem}

Note that Eq. (\ref{Eq:thmFE}) holds for $\alpha\in[0,\lambda)$ as well as $\alpha\in(-\infty,0)$. The exponent $\alpha$ needs to be strictly smaller than status packet arrival rate $\lambda$, or there is an unbounded average penalty. In this sense, the average penalty under an exponential penalty function with exponent $\alpha>0$ is highly sensitive to the status packet arrival rate $\lambda$.

\begin{theorem}
\label{thm:FS}
Under the FCFS discipline, with status packet arrivals and energy packet arrivals being Poisson processes of rate $\lambda=r\theta$ and $r$, and the service time being negligible, with penalty function $g\left(\Delta\right) = \mathbf{1}\left\{\Delta\geq\beta\right\}$, the average penalty is
\begin{eqnarray}
\label{Eq:SF}
C_\mathrm{FCFS,~step}= e^{-\lambda \beta} + \frac{e^{-r\beta}}{\theta^{-B}-\theta^{K+1}}\sum_{i=0}^{K}\frac{(\lambda\beta)^i}{i!} - \frac{\theta^{K+1}e^{-r\beta}}{\theta^{-B}-\theta^{K+1}}\sum_{i=0}^{K}\frac{(r\beta)^i}{i!} + \frac{e^{-\lambda \beta}\theta^{K+1}-e^{-r \beta}}{\theta^{-B}-\theta^{K+1}}.
\end{eqnarray}
\end{theorem}

The average penalty with penalty function $g\left(\Delta\right) = \mathbf{1}\left\{\Delta\geq\beta\right\}$ equals the probability that the AoI exceeding threshold $\beta$. In \cite{7415972}, the peak AoI is proposed as a metric that is suitable in the applications where the AoI is supposed to be lower than a given bound. However, comparing the CDF of peak AoI in Eq. (\ref{Eq:AF}) and AoI's threshold violation probability in Eq. (\ref{Eq:SF}), it is observed that with the same threshold $\beta$, we have$$\mathrm{P}\left\{A>\beta\right\} \geq C_\mathrm{FCFS,~step},$$which implies that the peak AoI violation probability does not equal the AoI violation probability. Therefore, the peak AoI violation probability can only serve as an upper bound of the AoI violation probability.

\subsection{LCFS Discipline}
\label{sec:L}
Under the LCFS discipline, the average penalty is expected to be smaller than its FCFS counterpart, since the latest status packet is delivered first. In a queuing system with the LCFS discipline, the probability distribution of sojourn time and peak AoI are obtained and summarized in the following lemmas:
\begin{lemma}
Under the LCFS discipline, with status packet arrivals and energy packet arrivals being Poisson processes of rate $\lambda$ and $r$, and the service time being negligible, the arrival rate $\tilde{\lambda}$ of valid updates is
\begin{equation}
\label{Eq:Llambda}
\tilde{\lambda}_\mathrm{LCFS} = \frac{1-\frac{\theta^{B+1}}{1+\theta} - \frac{\theta^{B+K+1}}{1+\theta}}{1-\theta^{B+K+1}}.
\end{equation}
\end{lemma}
\begin{proof}
See Appendix \ref{proof:l3}.
\end{proof}

Under the LCFS discipline when the buffer capacity $K>1$, outdated status packets is going to be delivered to the receiver unless they are pushed out of the buffer by fresher status packets. Therefore, the arrival rate of valid updates is smaller than the throughput of the system. The valid updates under the LCFS discipline consists of all status packets that arrive at the buffer when the battery is not empty, and those delivered before any other status packets arrive. It is observed from Eq. (\ref{Eq:Llambda}) that as the buffer capacity $K$ increases, there is a decrease in valid updates. The reason is that outdated status packets in the buffer is going to consume energy packets, which reduces the probability of a newly arrived status packet being a valid update.

\begin{lemma}
Under the LCFS discipline, with status packet arrivals and energy packet arrivals being Poisson processes of rate $\lambda=r\theta$ and $r$, and the service time being negligible, the CDFs of the sojourn time of valid updates and the peak AoI are expressions as Eq. (\ref{Eq:sojourn}) and (\ref{Eq:peak}).
\end{lemma}
\begin{proof}
See Appendix \ref{proof:l4}.
\end{proof}

\begin{figure*}[t]
\normalsize
\begin{eqnarray}
\label{Eq:sojourn}
\mathrm{Prob}\left[T \leq t\right] = 1-e^{-\left(\lambda+r\right)t}\frac{\left(1-\theta^{K+1}\right)}{\left(\theta^{-B}-1\right)\left(1+\theta\right)+\left(1-\theta^{K+1}\right)}
\end{eqnarray}
\begin{eqnarray}
\label{Eq:peak}
&~&\mathrm{Prob}\left[A\leq a\right] \notag\\
&=& 1 - e^{-\lambda a}\frac{\left(\theta^{-B}-\theta^{K+1}\right)\left(1+\theta\right)}{\left(\theta^{-B}-1\right)\left(1+\theta\right)+\left(1-\theta^{K+1}\right)} - e^{-ra}(1+\theta)\frac{K\theta^{-1} + \frac{\theta^{K+1} - 2 + \theta^{-1}}{1-\theta}}{\left(\theta^{-B}-1\right)\left(1+\theta\right)+\left(1-\theta^{K+1}\right)} \notag\\
&~& + \frac{e^{-\left(\lambda+r\right)a}(1+\theta)}{\left(\theta^{-B}-1\right)\left(1+\theta\right)+\left(1-\theta^{K+1}\right)}\left\{\frac{- 1 + \theta^{K+1}}{1+\theta} + (K\theta^{-1} + \theta^{-1} - \frac{\theta}{1-\theta})\sum_{k=0}^{K}\frac{\left(r\theta a\right)^{k}}{k!}\right.\notag\\
&~&\left. - \theta^{-1}\sum_{k=0}^{K-1}\frac{\left(r\theta a\right)^{k+1}}{k!} + \frac{\theta^{K+2}}{1-\theta}\sum_{k=0}^{K}\frac{\left(ra\right)^{k}}{k!}\right\}
\end{eqnarray}
\hrulefill
\vspace*{-0.2in}
\end{figure*}

For the penalty function being $g(\Delta) = \Delta$, $g(\Delta) = \alpha^{-1}\left(e^{\alpha\Delta}-1\right)$, and $g(\Delta) = \mathbf{1}\left\{\Delta-\beta\right\}$, the long-term average penalty under LFCS is obtained by Eq. (\ref{Eq:int}) and given in Theorem \ref{thm:LL}--\ref{thm:LS}.
\begin{theorem}
\label{thm:LL}
Under the LCFS discipline, with status packet arrivals and energy packet arrivals being Poisson processes of rate $\lambda=r\theta$ and $r$, and the service time being negligible, the average AoI is
\begin{equation}
\label{Eq:aveL}
C_\mathrm{LCFS,~linear} = \lambda^{-1} + \frac{r^{-1}}{\theta^{-B}-\theta^{K+1}}\left[ \frac{(1-\theta)\theta^{K+1}}{(1+\theta)^{K+1}} - \theta^{K+1} + \theta\right].
\end{equation}
\end{theorem}

\begin{theorem}
\label{thm:LE}
Under the LCFS discipline, with status packet arrivals and energy packet arrivals being Poisson processes of rate $\lambda=r\theta$ and $r$, and the service time being negligible, with penalty function $g\left(\Delta\right) = \alpha^{-1}\left(e^{\alpha\Delta}-1\right), \alpha < \lambda$, the average penalty is
\begin{eqnarray}
C_\mathrm{LCFS,~exp} =  \frac{1}{\lambda-\alpha} + \frac{1}{(r-\alpha)^2}\frac{\lambda}{\left(\theta^{-B}-\theta^{K+1}\right)}\left[1 + \left(\frac{\lambda}{\lambda+r-\alpha}\right)^{K+1}\frac{r-\lambda}{\lambda-\alpha} - \theta^{K+1}\frac{r-\alpha}{\lambda-\alpha}\right].
\end{eqnarray}
\end{theorem}

\begin{theorem}
\label{thm:LS}
Under the LCFS discipline, with status packet arrivals and energy packet arrivals being Poisson processes of rate $\lambda=r\theta$ and $r$, and the service time being negligible, with penalty function $g\left(\Delta\right) = \mathbf{1}\left\{\Delta\geq\beta\right\}$, the average penalty is
\begin{eqnarray}
C_\mathrm{LCFS,~step}&=& e^{-\lambda\beta} + \frac{e^{-(r+ \lambda)\beta}}{\theta^{-B}-\theta^{K+1}}\left\{e^{\lambda\beta}\frac{\theta^{K+2}}{1-\theta} + \left(K + \frac{1-2\theta}{1-\theta}\right)\left(e^{\lambda\beta} - \sum_{i=0}^{K}\frac{(\lambda\beta)^i}{i!}\right) \right.\notag\\
&~&\left.+ \lambda\beta\sum_{i=0}^{K-1}\frac{(\lambda\beta)^{i}}{i!} - \frac{\theta^{K+2}}{1-\theta}\sum_{i=0}^{K}\frac{(r\beta)^i}{i!}\right\}.
\end{eqnarray}
\end{theorem}

\subsection{Asymptotic Regime: $K\to\infty$}
When the buffer size is large enough, all the status packets is delivered to the receiver, and thus the throughput is $\lambda$. Next, we look into the average penalty when $K\to\infty$.
\begin{corollary}
\label{cor:1}
The average penalty of an energy harvesting status update system, with Poisson status packet arrivals and energy packet arrivals of rate $\lambda$ and $r = \lambda/\theta$, and an infinite-sized buffer, is summarized in TABLE \ref{tbl:1}.
\end{corollary}
\begin{table}[h]
\centering
\caption{Expressions of average penalty with three penalty functions under the FCFS and the LCFS disciplines when the buffer capacity $K\to\infty$. }
\vspace*{-0.1in}
\label{tbl:1}
\begin{tabular}{|c|c|c|}
\hline
Penalty functions&FCFS&LCFS\\
\hline
$g(\Delta)=\Delta$	&$\frac{1}{\lambda} + \frac{1}{\lambda} \frac{\theta^{B+2}}{1-\theta}$	&$\frac{1}{\lambda} + \frac{1}{\lambda}\theta^{B+2}$\\
\hline
$g(\Delta) = \alpha^{-1}\left(e^{\alpha\Delta}-1\right)$	&$\frac{1}{\lambda-\alpha} + r\alpha^{-1}\left(\frac{1-\theta}{r-\alpha-\lambda} - \frac{1}{r-\alpha}\right)\theta^{B}, ~\alpha<r-\lambda$	&$\frac{1}{\lambda-\alpha} + \frac{\lambda}{(r-\alpha)^2}\theta^B$\\
\hline
$g\left(\Delta\right) = \mathbf{1}\left\{\Delta\geq\beta\right\}$	&$e^{-\lambda \beta} + \left(e^{\lambda\beta} - 1\right)e^{-r \beta}\theta^B$ &$e^{-\lambda\beta} + \lambda\beta e^{-r\beta}\theta^B$\\
\hline
\end{tabular}
\vspace*{-0.2in}
\end{table}

As mentioned before, the difference in the average penalty between the energy harvesting case and the average-power-constrained case gradually diminishes with the growth of battery capacity $B$. When the $K\to\infty$, it is observed in TABLE \ref{tbl:1} that the average penalty decays exponentially to its limit with constant rate $\theta$ as the battery capacity $B$ increases, regardless of the service discipline or the penalty function.

\section{Non-Negligible Service Time Regime}
\label{sec:non-zero}

In \cite{8406846}, the authors analyze the average AoI of the model when the buffer capacity is $1$, and obtain the average AoI under the asymptotic region where the ratio between status packet arrival rate $\lambda$ and service rate $\mu$ goes to infinity. However, it is difficult to obtain the result when the service time is not negligible.

In this section, the service time of a status packet is assumed to be an i.i.d exponential random variable with mean $\mu^{-1}$. We show that the queues evolve as a QBD process, and compute the average peak AoI by matrix geometric method\cite{neuts1994matrix}. Analysis of the results in this section will be discussed at the beginning of Section V with illustrations.

\begin{figure}[h]
\vspace*{-0.2in}
\centering
\includegraphics[width=3in]{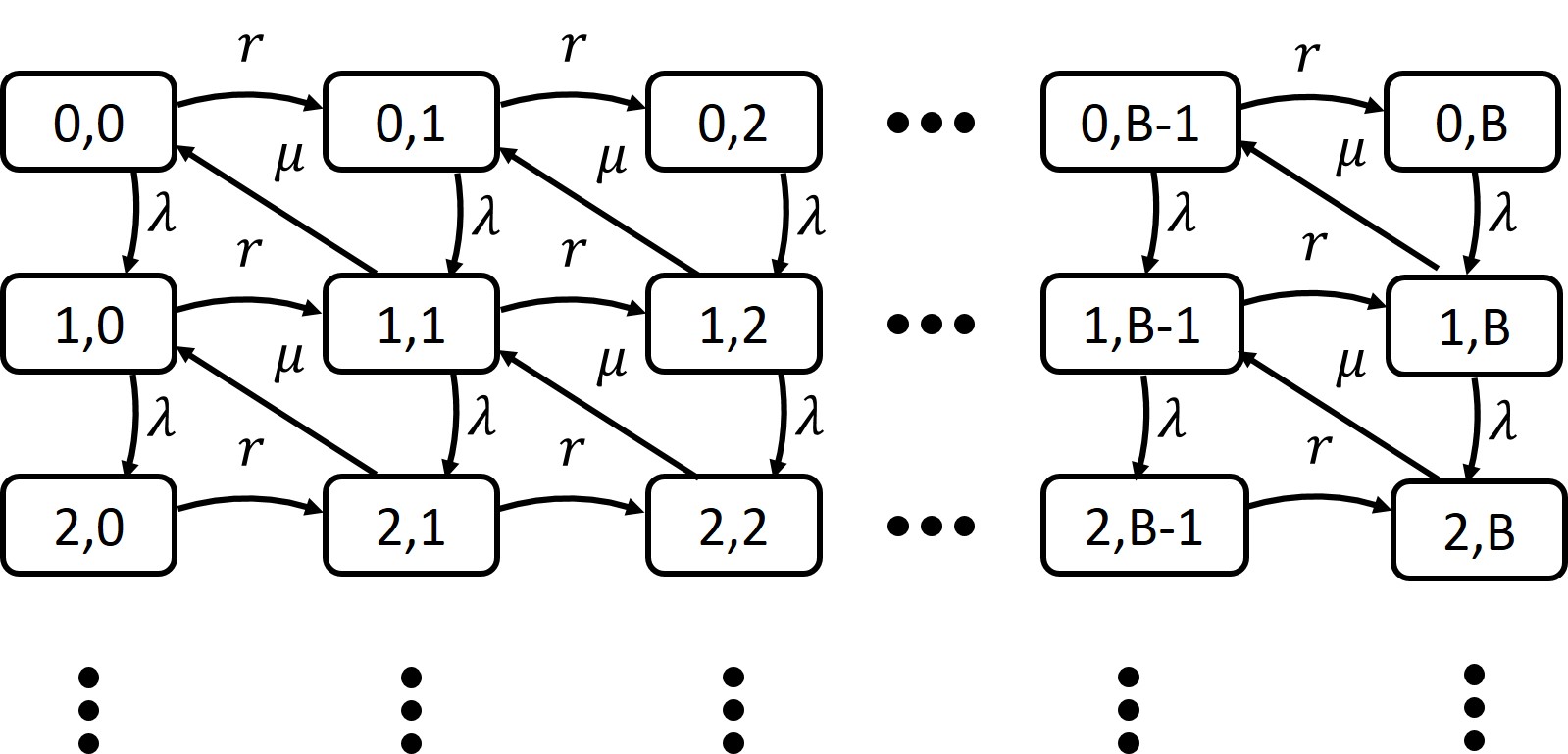}
\vspace*{-0.2in}
\caption{Two-dimensional state transition diagram for the general problem.}
\label{fig:two-dim-transit}
\vspace*{-0.1in}
\end{figure}

The states of the system are represented by tuple $\left(q_1, q_2\right)$. Note that when $q_1 >0$ and $q_2>0$, the server is serving a status packet with an energy packet until the status packet leaves the system, such that the state transits to $\left(q_1-1,q_2-1\right)$. At each state $\left(q_1, q_2\right)$, there are several possible transitions, summarized as follows:
\begin{enumerate}
\item From state $\left(q_1, q_2\right)$ to state $\left(q_1 + 1, q_2\right)$ with transition rate $\lambda$;
\item From state $\left(q_1, q_2\right)$ to state $\left(q_1, q_2+1\right)$ with transition rate $r$;
\item From state $\left(q_1, q_2\right)$ to state $\left(q_1-1, q_2-1\right)$ with transition rate $\mu$ if $\min\{q_1,q_2\}\geq1$.
\end{enumerate}

The corresponding Markov transition graph is illustrated in Fig. \ref{fig:two-dim-transit}. As the graph indicates, the two-dimensional Markov chain is a QBD process, of which the transition matrix is written as:
\begin{eqnarray*}
\mathbf{Q} = \left(
\begin{array}{ccccc}
\mathbf{\tilde{V}} & \mathbf{W}    & \mathbf{0}      & \mathbf{0}      & \cdots \\
\mathbf{U}         & \mathbf{V}    & \mathbf{W}    & \mathbf{0}      & \cdots \\
\mathbf{0}           & \mathbf{U}    & \mathbf{V}    & \mathbf{W}    & \cdots \\
\mathbf{0}           & \mathbf{0}     & \mathbf{U}   & \mathbf{V}    & \cdots \\
\vdots      & \vdots & \vdots & \ddots & \ddots \\
\end{array}
\right),
\end{eqnarray*}
in which $\mathbf{\tilde{V}}, \mathbf{U}, \mathbf{W}, \mathbf{V} \in \mathbb{R}^{\left(B+1\right)\times\left(B+1\right)}$ are
\begin{eqnarray*}
\mathbf{\tilde{V}} = -\lambda \mathbf{I} + \left(
\begin{array}{ccccc}
-r & r    & ~      & ~     & 0\\
~         & -r & r    & ~      & ~ \\
~           & ~ & \ddots    & \ddots    & ~ \\
~           & ~      & ~    & -r    & r \\
0      & ~ & ~ & ~ & 0 \\
\end{array}
\right),
\end{eqnarray*}
\begin{eqnarray*}
\mathbf{V} = -\lambda \mathbf{I} + \left(
\begin{array}{ccccc}
-r & r    & ~      & ~     & 0 \\
~         & -( r + \mu) & r    & ~      & ~ \\
~           & ~ & \ddots    & \ddots    & ~ \\
~           & ~      & ~    & -(r + \mu)    & r \\
0      & ~ & ~ & ~ & -\mu \\
\end{array}
\right),
\end{eqnarray*}
\begin{eqnarray*}
\mathbf{U} = \left(
\begin{array}{ccccc}
0~~~ & ~~~~    & ~~~~      & ~~~~     & 0~~~ \\
\mu         & 0    & ~      & ~ \\
~           & \mu & \ddots    & ~    & ~ \\
~           & ~      & \ddots    & 0    & ~ \\
0      & ~ & ~ & \mu & 0 \\
\end{array}
\right),
\end{eqnarray*}
and $\mathbf{W} = \lambda \mathbf{I}$.

Since the system is ergodic, there exists a stationary solution $\{\mathbf{p}_{i}\}$, where $\mathbf{p}_{i} = \left\{p_{i,0}, p_{i,1}, \cdots, p_{i,B}\right\}$ denotes the probability of the system being at state $(i,\cdot)$, such that $\mathbf{p}_{i}$ satisfies the following recursive relationship
\begin{equation}
\label{Eq:seq}
\mathbf{p}_{{i+1}} = \mathbf{p}_{i}\mathbf{R}, i \in \mathbb{N}.
\end{equation}
Since
\begin{equation}
\label{Eq:lim}
\left(\mathbf{p}_{0},\mathbf{p}_{1}, \cdots\right)\mathbf{Q} = \mathbf{0},
\end{equation}
by substituting (\ref{Eq:seq}) into (\ref{Eq:lim}), we have
\begin{equation}
\mathbf{R}^2\mathbf{U} + \mathbf{R}\mathbf{V} + \mathbf{W} = \mathbf{0}
\label{Eq:sqr}
\end{equation}
and
$\mathbf{p}_{0}\left(\mathbf{\tilde{V}} + \mathbf{R}\mathbf{U}\right) = \mathbf{0}$,
along with constraint
$\mathbf{p}_0\sum_{i=0}^{\infty}\mathbf{R}^i\mathbf{1}^\mathrm{T} = \mathbf{p}_0\left(\mathbf{I}-\mathbf{R}\right)^{-1}\mathbf{1}^\mathrm{T} = 1$,
in which $\mathbf{1} = (1,\cdots,1)\in\mathbb{R}^{B+1}$. Based on (\ref{Eq:sqr}), by iteratively computing $\mathbf{R}_{n+1} = -(\mathbf{R}^2_n\mathbf{U}+\mathbf{W})\mathbf{V}^{-1}$ with initial condition $\mathbf{R}_0 = \mathbf{0}$, $\mathbf{R}$ can be approached, and the stationary distribution of system states is obtained. The algorithm to compute the stationary distribution is summarize in Algorithm \ref{alg:1}.

\begin{algorithm}[h]
\caption{Matrix geometric method to obtain the stationary distribution of the non-negligible service time problem}
\label{alg:1}
\begin{algorithmic}[1]
\label{alg: LOKI}
\REQUIRE $\lambda, \mu, r, B,\epsilon = 10^{-8}$;
\STATE $\mathbf{R}_0 = \mathbf{0}$;
\REPEAT
\STATE $\mathbf{R}_{n+1} = -(\mathbf{R}^2_n\mathbf{U}+\mathbf{W})\mathbf{V}^{-1}$;
\UNTIL $\max_{i,j}|(\mathbf{R}_{n+1}-\mathbf{R}_{n})_{ij}|<\epsilon$;
\STATE Find the eigenvector $\mathbf{p}_0$ of $(\mathbf{\tilde{V}} + \mathbf{R}\mathbf{U})$ with eigenvalue $0$ and  satisfying $\mathbf{p}_0\left(\mathbf{I}-\mathbf{R}\right)^{-1}\mathbf{1}^\mathrm{T} = 1$.
\end{algorithmic}
\end{algorithm}

According to Little's law, the mean sojourn time of a status packet is
$$\mathbb{E}\left[T\right]  =  \frac{\bar{q_1}}{\lambda}= \frac{1}{\lambda}\sum_{i=1}^{+\infty}i\mathbf{p}_{0}\mathbf{R}^i\mathbf{1}^\mathrm{T} =  \frac{1}{\lambda}\mathbf{p}_{0}(\mathbf{I}-\mathbf{R})^{-2}\mathbf{R}\mathbf{1}^\mathrm{T}. $$With the stationary distribution $\{\mathbf{p}_0\mathbf{R}^i\}_{i\geq0}$, the expression for average peak AoI is obtained by Eq. (\ref{Eq:page}), and summarized in the following theorem:
\begin{theorem}
Under the FCFS discipline, with status packet arrivals and energy packet arrivals being Poisson processes of rate $\lambda=r\theta$ and $r$, and the service time being i.i.d. random variable following exponential distribution with mean $\mu^{-1}$, the average peak AoI is
\begin{eqnarray}
\mathbb{E}\left[A\right] = \frac{1}{\lambda} + \frac{1}{\lambda}\mathbf{p}_{0}(\mathbf{I}-\mathbf{R})^{-2}\mathbf{R}\mathbf{1}^\mathrm{T},
\end{eqnarray}
where $\mathbf{p}_{0}$ and $\mathbf{R}$ are computed by Algorithm \ref{alg:1}.
\end{theorem}

\section{Numerical Analysis}
\label{sec:anal}

Fig. \ref{fig:B4} illustrates the average peak AoI versus the ratio of status packet arrival rate to service rate, and the ratio of energy arrival rate to service rate, respectively. The service rate $\mu$ is set to $1\mathrm{s}^{-1}$, and the battery capacity $B$ is set to $5$ units. As the figure shows, the average peak age first decreases then rises with the growth of $\frac{\lambda}{\mu}$. The reason is intuitive: When the status packet arrival rate is small, the lack of status information generation leads to a low update frequency, while when the status packet arrival rate is large, long queuing delay becomes the dominant factor that aggravates the freshness of status information. Therefore, there exists an optimal status update frequency that achieves the minimum average peak AoI, which infers the optimal sensing rate for remote status update. It is also observed from the figure that the energy arrival rate $r$ should be larger than the status packet arrival rate $\lambda$. Otherwise, the status packet queue is not stable.

\begin{figure}[t]
\vspace*{-0.2in}
\centering
\includegraphics[width=3.3in]{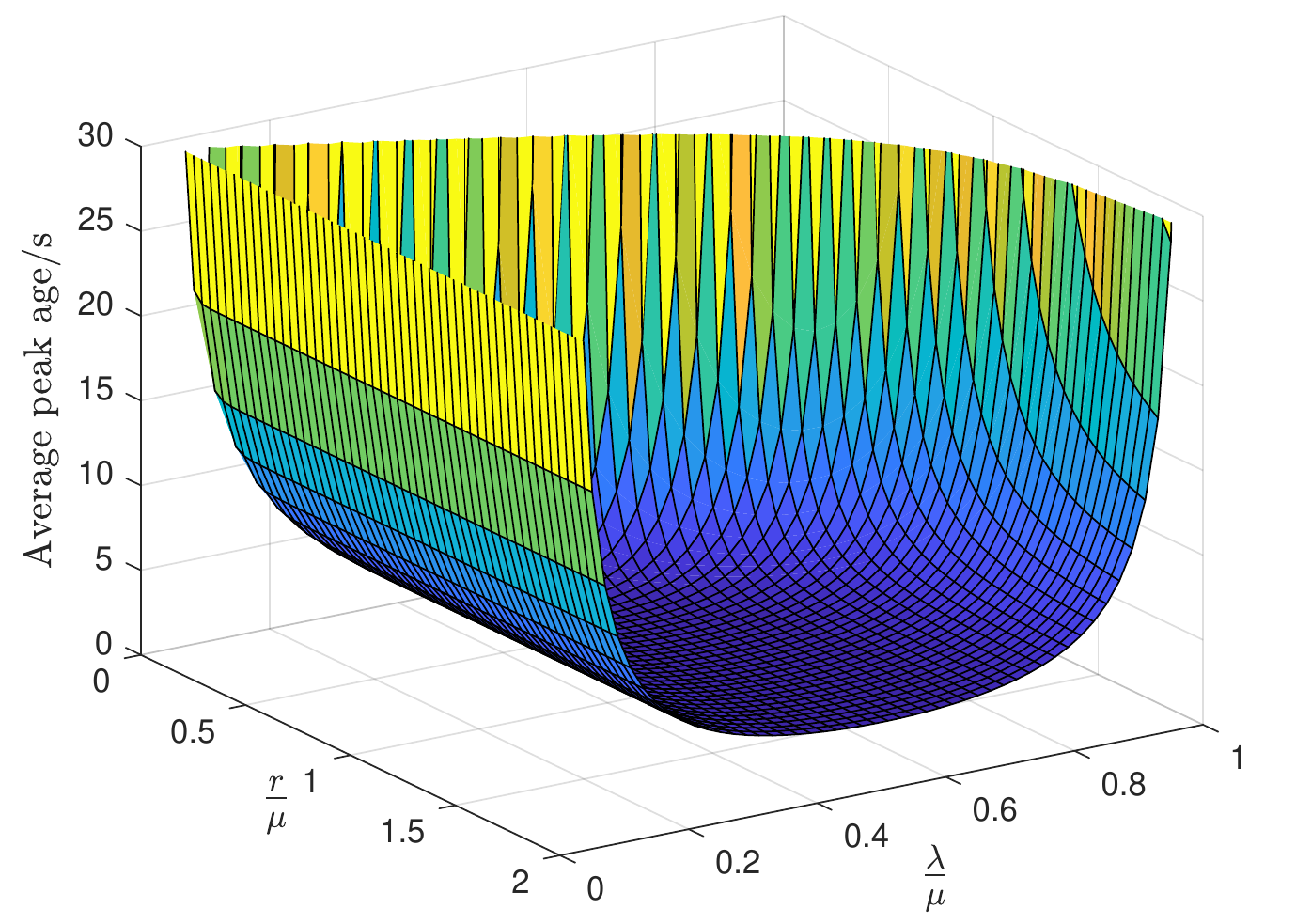}
\vspace*{-0.1in}
\caption{The average peak AoI given that $B=5$ and $\mu = 1\mathrm{s}^{-1}$. }
\label{fig:B4}
\vspace*{-0.2in}
\end{figure}

Next, we focus on the negligible-service-time regime to get insights into the problem.

\begin{figure}[h]
\centering
\includegraphics[width=3in]{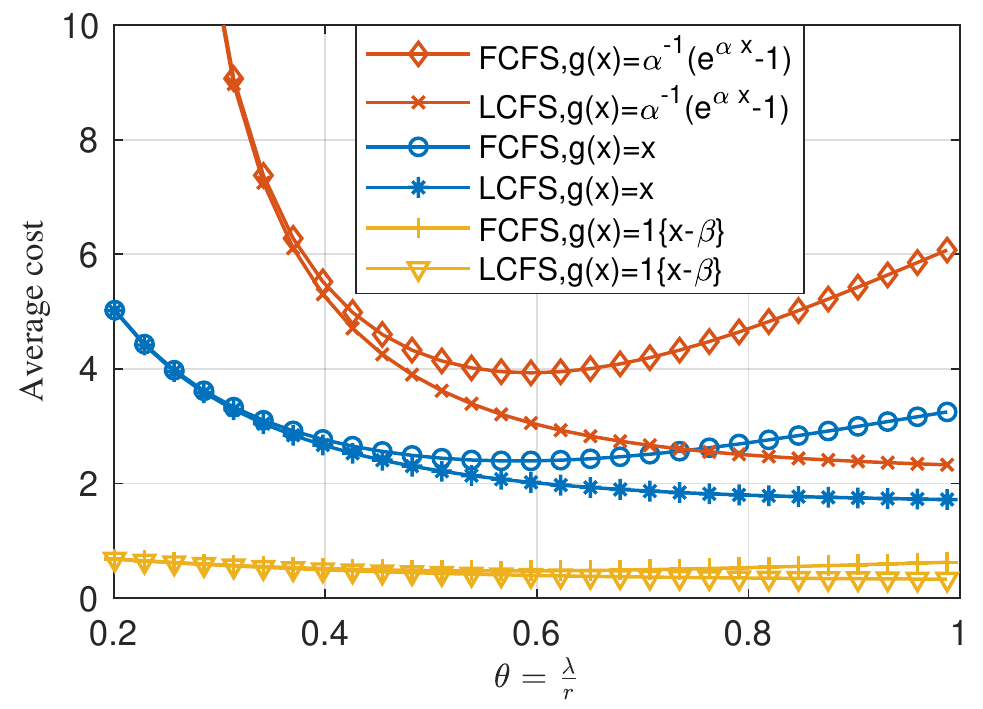}
\vspace*{-0.2in}
\caption{Comparison of average penalty in a status update system with battery size $B=1$, buffer capacity $K=5$, exponent $\alpha = 0.2$, threshold $\beta = 2$, and energy arrival rate $r=1\mathrm{s}^{-1}$. }
\label{fig:all}
\vspace*{-0.2in}
\end{figure}

Fig. \ref{fig:all} plots the three types of average penalty under the FCFS and LCFS disciplines. The average penalty under exponential penalty function is more sensitive to the change in data-to-energy ratio $\theta$, and the average penalty of shifted unit step function, \textit{i.e.}, the violation probability, is the least sensitive. The overall characteristics of the three penalty functions have several similarities:
\begin{enumerate}
\item Under the FCFS discipline, the average penalty first decreases then increases with the growth of $\theta$. This result is consistent with the one in average peak AoI when service time cannot be neglected.
\item Under the LCFS discipline, a larger $\theta$ always gives a lower penalty, owing to a larger update frequency.
\item The LCFS discipline always outperforms the FCFS discipline. This result is consistent with the work in \cite{7541763}.
\end{enumerate}

\begin{figure}[h]
\vspace*{-0.2in}
  \centering
  \subfigure[Average AoI]{
    \label{fig:cmpK}
    \includegraphics[width=2.8in]{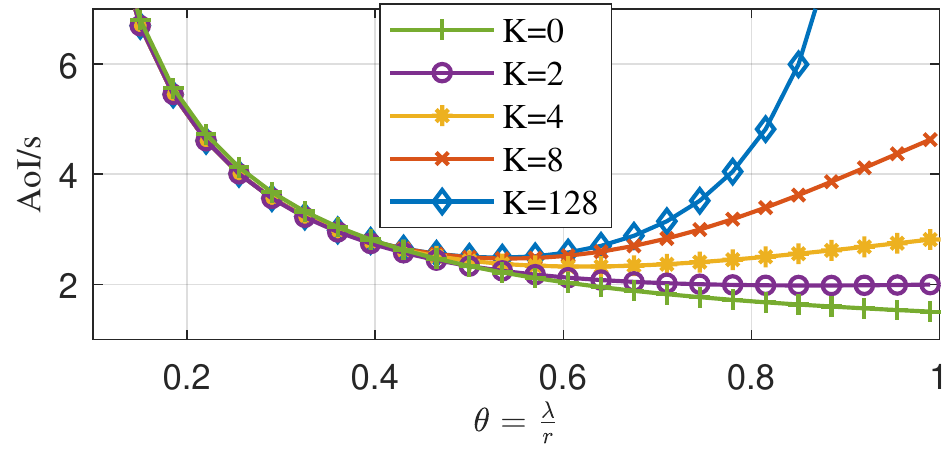}}
    \subfigure[The regime where $K=1$ gives the highest AoI.]{
   \label{fig:cmpKl}
   \includegraphics[width=2.8in]{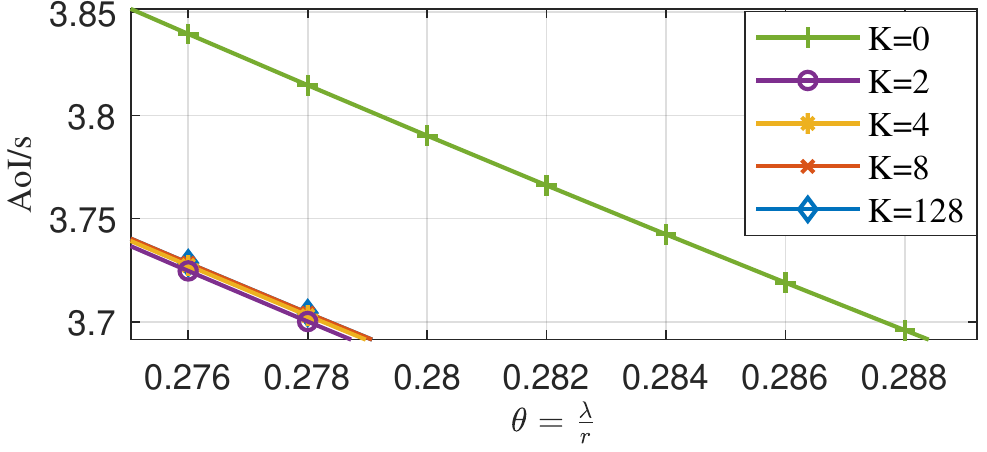}}
  \subfigure[Average penalty with exponential penalty function $g\left(\Delta\right) = \alpha^{-1}\left(e^{\alpha\Delta}-1\right)$ and exponent $\alpha = 0.2$.]{
    \label{fig:cmpK_exp}
    \includegraphics[width=2.8in]{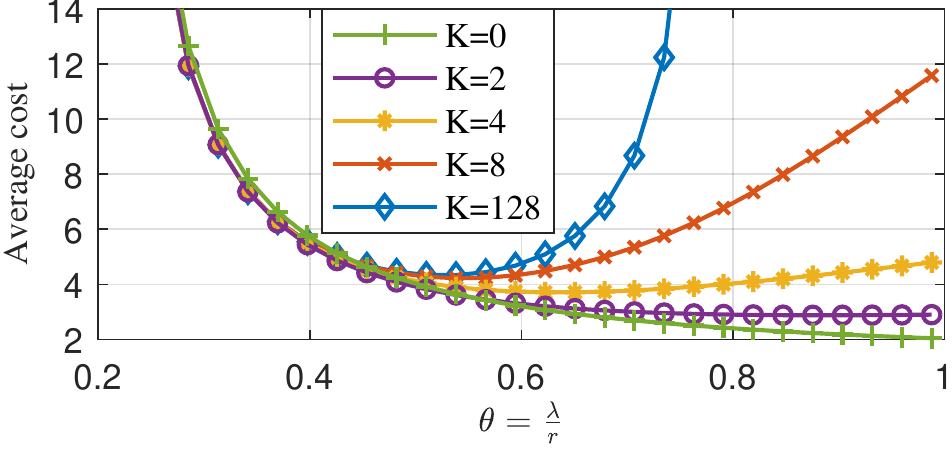}}
\subfigure[Violation probability under threshold $\beta=5$.]{
    \label{fig:cmpK_s}
    \includegraphics[width=2.8in]{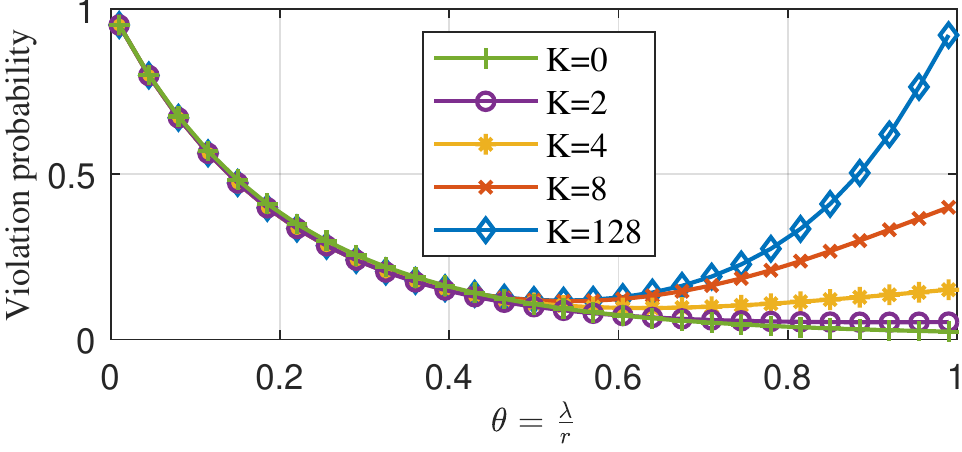}}
  \caption{Average penalty under FCFS discipline with energy arrival rate $r=1\mathrm{s}^{-1}$ and battery capacity $B=1$. }
  \label{fig:K}
  \vspace*{-0.3in}
\end{figure}

\subsection{Impact from Buffer Capacity $K$ under FCFS}
Fig. \ref{fig:cmpK} depicts the average AoI under different buffer capacities and ratio $\theta$ with constant energy arrival rate. It is shown that the average AoI is significantly lower, if newly arrived status packets are dropped when the status packet arrival rate is large. However, as observed in Fig. \ref{fig:cmpKl}, blocking more status packets cannot always reduce the average AoI when status packet arrival rate is small. The optimal buffer capacity under given battery capacity and status packet arrival rate can be found by numerical methods. Illustrations of average penalties under exponential penalty function and violation probability are shown in Fig. \ref{fig:cmpK_exp} and Fig. \ref{fig:cmpK_s}.

Moreover, as shown in Fig. \ref{fig:cmpK_exp}, the average penalty under the exponential penalty function is extremely sensitive to status packet arrival rate when the buffer capacity is large. The sensitivity of exponential penalty function is previously observed in \cite{8006543} for $K=\infty$ and in the packet management problem in \cite{8000687}. Especially, when the buffer capacity goes to infinity, the average penalty exists only if $\lambda < r-\alpha$. The reason is the under the exponential penalty function, the penalty rises sharply when AoI is large. However, the upper bound condition on status packet arrival rate is not necessary under the LCFS discipline.

\begin{figure}[h]
  \centering
    \vspace*{-0.1in}
  \subfigure[Average AoI]{
    \label{fig:cmpB}
    \includegraphics[width=2.8in]{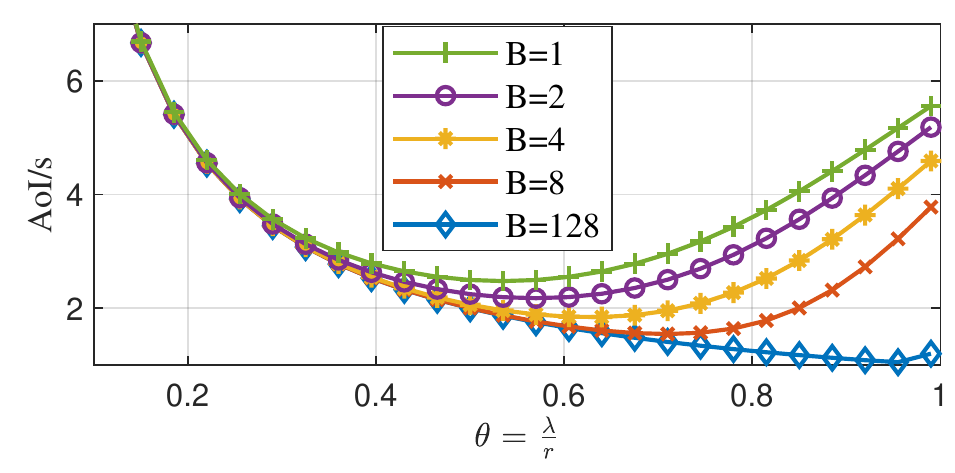}}
  \subfigure[Average penalty with exponential penalty function $g\left(\Delta\right) = \alpha^{-1}\left(e^{\alpha\Delta}-1\right)$ and exponent $\alpha = 0.2$.]{
    \label{fig:cmpB_exp}
    \includegraphics[width=2.8in]{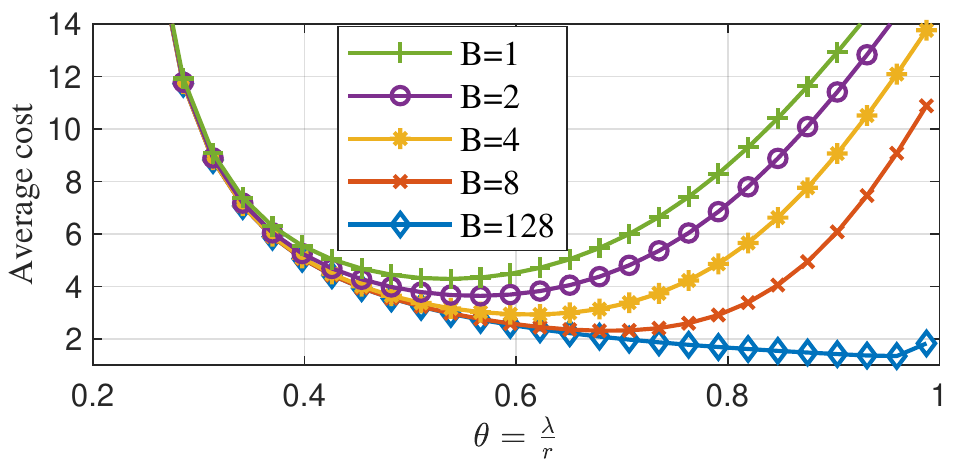}}
\subfigure[Violation probability under threshold $\beta = 5$.]{
    \label{fig:cmpB_s}
    \includegraphics[width=2.8in]{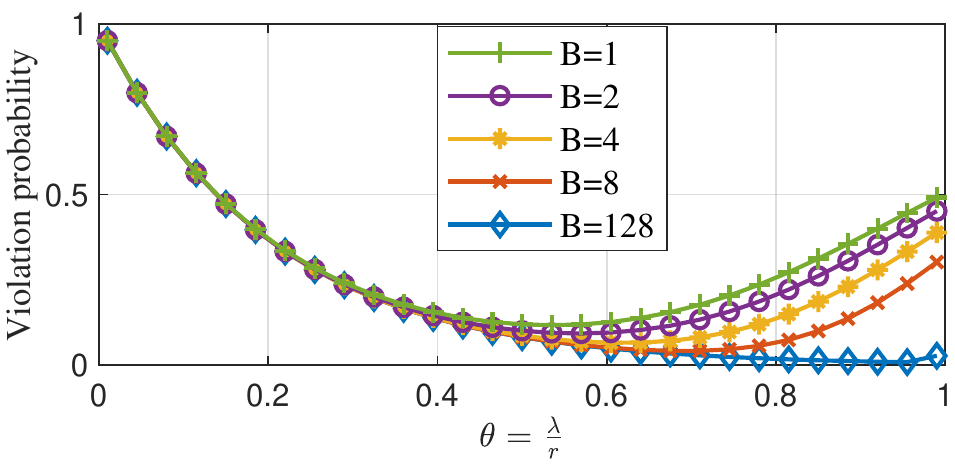}}
  \caption{Average penalty under FCFS discipline with energy arrival rate $r=1\mathrm{s}^{-1}$ and buffer capacity $K=10$. }
  \label{fig:B} %
    \vspace*{-0.3in}
\end{figure}

\subsection{Impact from Battery Capacity $B$ under FCFS}
Fig. \ref{fig:B} compares average penalty under different battery capacities and ratio $\theta$ with constant energy arrival rate $r=1\mathrm{s}^{-1}$. As shown in the figure, when $\theta$ is small, the increase in battery capacity does not noticeably reduce the average penalty. As the status packet arrival rate grows, different from the average power constrained case, the average penalty first drops then increases. The optimal status packet arrival rate can be found by bisection method.
\begin{figure}[h]
\centering
\includegraphics[width=2.8in]{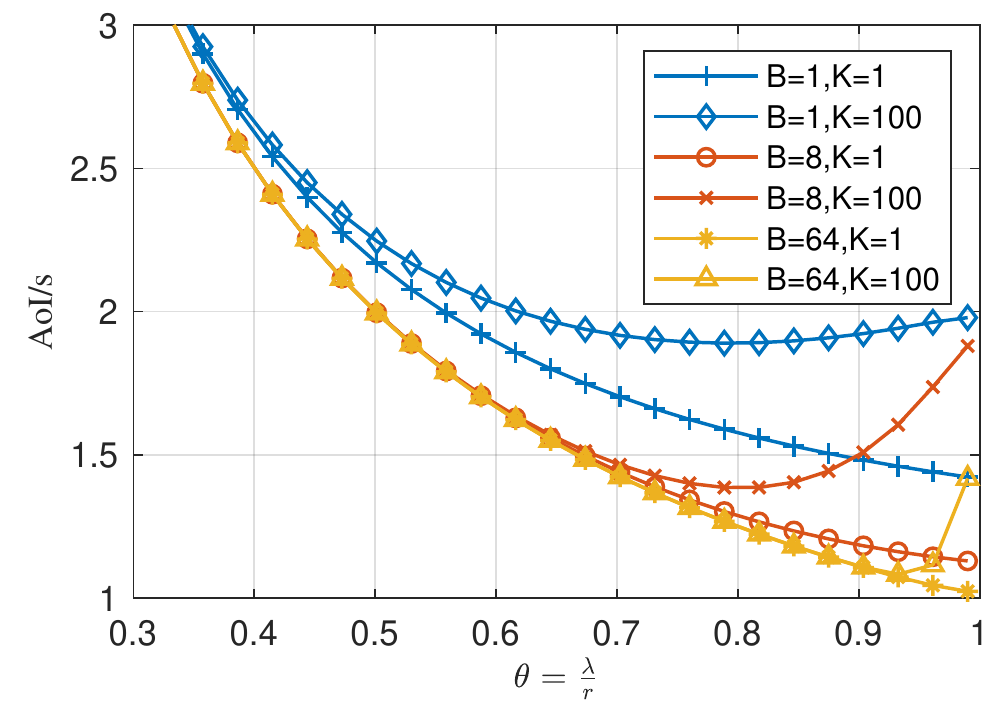}
\vspace*{-0.2in}
\caption{Average AoI of a LCFS system under different buffer capacitys and ratio $\theta$ with energy arrival rate $r=1\mathrm{s}^{-1}$. }
\label{fig:L}
\vspace*{-0.2in}
\end{figure}
\subsection{Impact from the LCFS Discipline}
According to Theorem 4--6, the average penalty under LCFS discipline also decays as battery capacity increases. Fig. \ref{fig:L} plots the average AoI under LCFS discipline, with data buffer $K=1$ and $K=100$. When $K=100$, little status packets is dropped, and the average AoI first drops and then increases as the status packet arrival rate grows, while with $K=1$, \textit{i.e.}, only the latest status packet is kept in the queue at each time, a higher status update frequency always gives a lower average age. Additionally, it is noticed from the figure and can be easily proved that LCFS discipline with $K=1$ leads to the lowest average AoI if the server is working-conserving. The same result also holds in the cases with exponential penalty function and shifted unit step function.

\section{Conclusions and Future Work}
\label{sec:con}
In this paper, we investigate the non-linear AoI for a status update system powered by renewable energy sources. With packet generation and energy arrival both being a Poisson process and service time following exponential distribution, a method to compute the average peak AoI is proposed. To gain more insights into the problem, a special case in which the service time is small enough is analyzed. The closed-form expressions of average penalty under linear penalty function, exponential penalty function, and shifted unit step function for both FCFS and LCFS systems are obtained. Results show that under the FCFS discipline, the average penalty under the exponential penalty function is extremely sensitive to the status information generation frequency, especially when the buffer capacity or the parameter $\alpha$ in exponent is large. As battery capacity increases, the difference in the average penalty between the energy-harvesting case and the average-power-constrained case exponentially decays with rate equal to the ratio of status information generation frequency to energy arrival rate. For FCFS systems, there exists a unique minimum average penalty that can be achieved by a proper status information generation frequency, while in a LCFS system with unit buffer capacity, the minimum penalty is achieved, and a larger status information generation frequency always gives a smaller average penalty. It is also noticed that blocking more status packets can reduce the average penalty when the status generation frequency is large. The performance of the LCFS discipline and finite buffer capacity in a system with non-negligible service time needs to be investigated in future work.

\begin{appendices}
\section{Proof of Lemma 1}
\label{proof:l1}
Under FCFS, a data  packet is a valid update \textit{if and only if} the data buffer is not full upon its arrival. Since the queuing system is equivalent to an M/M/1 queue, we have $$\lim_{n\to\infty}\mathrm{P}\left\{S_{n}=j\right\} = \frac{\left(1-\theta\right)\theta^{j+B}}{1-\theta^{K+B+1}}. $$According to the {\textit {PASTA}} (Poisson Arrivals See Time Average) property, we have
\begin{equation}
\tilde{\lambda}_\mathrm{FCFS} = \lambda\sum_{j=-B}^{K-1}\lim_{n\to\infty}\mathrm{P}\left\{S_{n}=j\right\} = \lambda \frac{1-\theta^{K+B}}{1-\theta^{K+B+1}}.
\end{equation}

\section{Proof of Lemma 2}
\label{proof:l2}
By the total probability formula, we have that
$$\Prob{A_n\leq a} = \Prob{X_n+T_n\leq a} = \int_0^\infty\Prob{T_n\leq a-x\big|X_n = x}f_{X_n}(x)\mathrm{d}x, $$
and that
$$\Prob{T_n\leq t} = \int_0^\infty\Prob{T_n\leq t\big|X_n = x}f_{X_n}(x)\mathrm{d}x. $$
Therefore, we first obtain the CDF of inter-arrival time $X$ of valid updates, and the conditional probability distribution of sojourn time $T$ given inter-arrival time $X$.

Denoting the system state right before the arrival of the $n$-th valid update as $S_n^-$, since valid updates are the status packets that enter the data buffer, we get $S_n^- \leq K-1$. Applying the total probability formula again, we have
\begin{eqnarray}
\mathrm{P}\left\{T_n\leq t|X_n\right\} &=& \sum_{i=-B}^{K-1}\mathrm{P}\left\{S_n^-=i|X_n\right\}\mathrm{P}\left\{T_n\leq t|S_n^-=i,X_n\right\}\notag\\
& \overset{(a)}{=}&  \sum_{i=-B}^{K-1}\mathrm{P}\left\{S_n^-=i|X_n\right\}\mathrm{P}\left\{T_n\leq t|S_n^-=i\right\},
\label{Eq:cond}
\end{eqnarray}
where (a) uses the fact that the sojourn time $T_n$ of the $n$-th valid update and the inter-arrival time $X_n$ between the $(n-1)$-th and the $n$-th valid update are conditionally independent given $S_n^-$.

The second term on the right-hand side of Eq. (\ref{Eq:cond}) is expressed as Eq. (\ref{Eq:ti}): when the battery is empty and there are already $i$ status packets waiting to be served, \textit{i.e.}, $S_n^-\geq i$, the new status packets will not be served until the arrival of $i+1$ energy packets; when the battery is not empty ($S_n^-<0$), the new status packets will be delivered instantly.
\begin{eqnarray}
\mathrm{P}\left\{T_n\leq t|S_n^-=i\right\} = \left\{
\begin{aligned}
&\sum_{j=i+1}^{+\infty}\frac{\left(rt\right)^j}{j!}e^{-rt}, &\mathrm{~if~} i\geq 0;\\
&1, &\mathrm{~if~} i < 0.
\end{aligned}
\right.
\label{Eq:ti}
\end{eqnarray}

Using the total probability formula in the first term on the right-hand side of Eq. (\ref{Eq:cond}) , we have
\begin{gather}
\mathrm{P}\left\{S_n^-=i|X_n\right\} =  \sum_{j=\max\left\{i-1,-B\right\}}^{K-1}\mathrm{P}\left\{S_{n-1}^-=j|X_n\right\}\mathrm{P}\left\{S_n^-=i|X_n,S_{n-1}^-=j\right\},
 \label{Eq:tp}
\end{gather}
where uses the condition that $S_{n-1}^- + 1 \geq S_{n}^-$ since there are only one valid update between $S_{n-1}^-$ and $S_{n}^-$.

To get Eq. (\ref{Eq:tp}), we first obtain the CDFs of inter-arrival time $X_n$ and system state $S_n^-$ before arrival. According to the {\textit {PASTA}} property, $\mathrm{P}\left\{S_{n}^-=i\right\}$ approaches the stationary probability of the system being at state $i$  when $n$ goes to infinity given the condition that the state is not $K$. Since the queuing system is equivalent to an M/M/1 queue, which gives $\lim_{n\to\infty}\mathrm{P}\left\{S_{n}=i\right\} = \frac{\left(1-\theta\right)\theta^{i+B}}{1-\theta^{K+B+1}}$ the probability of state being $i$ is
\begin{eqnarray}
\lim_{n\to\infty}\mathrm{P}\left\{S_{n}^-=i\right\}=\frac{\lim_{n\to\infty}\mathrm{P}\left\{S_{n}=i\right\}}{1-\lim_{n\to\infty}\mathrm{P}\left\{S_{n}=K\right\}}
= \frac{1-\theta}{1-\theta^{K+B}}\theta^{i+B}.
\label{Eq:MM1}
\end{eqnarray}
Next, we apply the total probability formula to obtain the CDF of inter-arrival time $X_n$:
\begin{eqnarray}
\label{Eq:x00}
\mathrm{P}\left\{X_n\leq x\right\} &=& \mathrm{P}\left\{X_n\leq x|S_{n-1}^-=K-1\right\}\mathrm{P}\left\{S_{n-1}^-=K-1\right\} \notag\\
&~&+ \sum_{j = -B}^{K-2}\mathrm{P}\left\{X_n\leq x|S_{n-1}^-=j\right\}\mathrm{P}\left\{S_{n-1}^-=j\right\}.
\end{eqnarray}
When $S_{n-1}^- < K-1$, the data buffer cannot be full at the arrival of the next status packet. Therefore, any status packet arrives right after the $(n-1)$-th valid update is a valid update, which gives
\begin{equation}
\label{Eq:xk1}
\mathrm{P}\left\{X_n\leq x|S_{n-1}^-=j\right\} = 1-e^{-\lambda x}, j \leq K-2;
\end{equation}
When $S_{n-1}^- = K-1$, there are at least one energy packet arrival before the next status packet entering the buffer. Thus, we have
\begin{eqnarray}
\label{Eq:xk2}
\mathrm{P}\left\{X_n\leq x|S_{n-1}^-=K-1\right\} 
&=& 1-\frac{1}{1-\theta}e^{-\lambda x} + \frac{\theta}{1-\theta}e^{-rx}.
\end{eqnarray}
Substituting Eq. (\ref{Eq:MM1}), (\ref{Eq:xk1}) and (\ref{Eq:xk2}) into the right-hand side of Eq. (\ref{Eq:x00}), we have
\begin{equation}
\label{Eq:x}
\mathrm{P}\left\{X_n\leq x\right\}=1+e^{-rx}\frac{\theta^{K+B}}{1-\theta^{K+B}}-e^{-\lambda x}\frac{1}{1-\theta^{K+B}}.
\end{equation}

Next we obtain the first term on the right-hand side of Eq. (\ref{Eq:tp}) by Bayes' theorem:
\begin{eqnarray*}
\mathrm{P}\left\{S_{n-1}^-=j|X_n=x\right\} = \frac{f_{X_n|S_{n-1}^-=j}(x)\mathrm{P}\left\{S_{n-1}^-=j\right\}}{f_{X_n}(x)}.
\end{eqnarray*}
With Eq. (\ref{Eq:MM1}), (\ref{Eq:xk1})--(\ref{Eq:x}), the above equation can be expressed as:
\begin{eqnarray}
\label{Eq:sx}
\mathrm{P}\left\{S_{n-1}^-=j|X_n=x\right\} =\left\{
\begin{aligned}
&\frac{\left(1-\theta\right)\theta^{j+B+1} e^{-\lambda x}}{-e^{-rx}\theta^{K+B}+ \theta e^{-\lambda x}}, &\mathrm{~if~} j < K-1;\\
&\frac{-e^{-rx}\theta^{K+B}+e^{-\lambda x}\theta^{K+B}}{-e^{-rx}\theta^{K+B}+\theta e^{-\lambda x}},& \mathrm{~if~} j = K-1.
\end{aligned}
\right.
\end{eqnarray}

Now we obtain the second term on the right-hand side of Eq. (\ref{Eq:tp}). If $S_{n-1}^-\leq K-2$ and  $S_n^->-B$, then there are $(S_{n-1}^- - S_{n}^- + 1)$ energy packets arrived between the arrivals of the $(n-1)$-th and the $n$-th valid update, which gives
\begin{equation}
\mathrm{P}\left\{S_n^-=i|X_n=x,S_{n-1}^-=j\right\} = \frac{\left(rx\right)^{j+1-i}}{\left(j+1-i\right)!}e^{-rx}, \mathrm{~if~}j\leq K-2, i>-B.
\label{Eq:psn1}
\end{equation}
If $S_{n-1}^-\leq K-2$ and  $S_n^-=-B$, then there are at least $(S_{n-1}^- + B + 1)$ energy packets arrived between the arrivals of the $(n-1)$-th and the $n$-th valid update:
\begin{align}
\mathrm{P}\left\{S_n^-=i|X_n=x,S_{n-1}^-=j\right\} = \sum_{m=j+B+1}^{+\infty}\frac{\left(rx\right)^{m}}{m!}e^{-rx}, \mathrm{~if~}j\leq K-2, i=-B.
\label{Eq:psn2}
\end{align}
If $S_{n-1}^-= K-1$, by the definition, we have
\begin{eqnarray}
\label{Eq:sxs1}
\mathrm{P}\left\{S_n^-=i|X_n=x,S_{n-1}^-=K-1\right\} = \frac{f_{S_n^-,X_n|S_{n-1}^-=K-1}(i,x)}{f_{X_n|S_{n-1}^-=K-1}(x)}.
\end{eqnarray}
Further, if $S_n^->-B$, then what happens between the arrivals of the $(n-1)$-th and the $n$-th valid update is that the first energy packet arrives at $t$ seconds after the arrival of the $(n-1)$-th valid update, and that there are $(K-1-S_n^-)$ energy packets and no status packet arrived between the arrivals of the first energy packet and the $n$-th valid update. Thus, we have
\begin{eqnarray}
\label{Eq:sxs2}
f_{S_n^-,X_n|S_{n-1}^-=K-1}(i,x) = \int_0^xre^{-r(x-t)}\frac{(rt)^{K-i-1}}{(K-i-1)!}e^{-rt}\lambda e^{-\lambda t}\,\mathrm{d}t, \mathrm{~if~}S_n^->-B.
\end{eqnarray}
If $S_n^-=-B$, then between the arrivals of the $(n-1)$-th and the $n$-th valid update, the first energy packet arrives at $t$ seconds after the arrival of the $(n-1)$-th valid update, and there are at least $K+B-1$ energy packets and no status packet arrived between the arrivals of the first energy packet and the $n$-th valid update:
\begin{eqnarray}
\label{Eq:sxs3}
f_{S_n^-,X_n|S_{n-1}^-=K-1}\left(-B, x\right) = \int_0^xre^{-r(x-t)}\sum_{l=K+B-1}^{+\infty}\frac{(rt)^l}{l!}e^{-rt}\lambda e^{-\lambda t}\,\mathrm{d}t.
\end{eqnarray}
Substituting Eq. (\ref{Eq:xk2}), (\ref{Eq:sxs2}) and (\ref{Eq:sxs3}) into Eq. (\ref{Eq:sxs1}), we get
\begin{eqnarray}
\label{Eq:psn3}
&~&\mathrm{P}\left\{S_n^-=i|X_n=x,S_{n-1}^-=K-1\right\}\notag\\
&=& \left\{
\begin{aligned}
&\frac{1-\theta}{\theta^{K-i}}\frac{e^{-rx}}{e^{-\lambda x} - e^{-rx}}\sum_{m=K-i}^{+\infty}\frac{(\lambda x)^m}{m!}e^{-\lambda x}, &\mathrm{~if~}S_n^->-B;\\
&\frac{e^{-r(1+\theta)x}}{e^{-\lambda x} - e^{-rx}}\left(\sum_{i=K+B}^{+\infty}\frac{(rx)^i}{i!}-\theta^{-K-B+1}\sum_{i=K+B}^{+\infty}\frac{(\lambda x)^i}{i!}\right), &\mathrm{~if~}S_n^-=-B.
\end{aligned}
\right.
\end{eqnarray}

Substituting Eq. (\ref{Eq:sx}), (\ref{Eq:psn1}), (\ref{Eq:psn2}) and (\ref{Eq:psn3}) into Eq. (\ref{Eq:tp}), we get
\begin{eqnarray}
\mathrm{P}\left\{S_n^-=i|X_n=x\right\} = \left\{
\begin{aligned}
&\frac{(1-\theta)\theta^{B+i}e^{-rx}}{-e^{-rx}\theta^{K+B}+\theta e^{-\lambda x}}, &\mathrm{~if~}i  > -B;\\
&\frac{(e^{-\lambda x}-e^{-rx})\theta}{-e^{-rx}\theta^{K+B}+\theta e^{-\lambda x}}, &\mathrm{~if~}i  = -B.
\end{aligned}
\right.
\label{Eq:ix}
\end{eqnarray}

Combining Eq. (\ref{Eq:ix}) with Eq. (\ref{Eq:ti}), Eq. (\ref{Eq:cond}) becomes
\begin{eqnarray*}
\label{Eq:tx}
\mathrm{P}\left\{T_n\leq t|X_n=x\right\} = 1 &-& \frac{\theta^{B}e^{-rx}}{-e^{-rx}\theta^{K+B}+\theta e^{-\lambda x}}\sum_{n=0}^{K-1}\frac{\left(\lambda t\right)^n}{n!}e^{-rt}\notag\\
 &+& \frac{\theta^{K+B}e^{-rx}}{-e^{-rx}\theta^{K+B}+\theta e^{-\lambda x}}\sum_{n=0}^{K-1}\frac{\left(rt\right)^n}{n!}e^{-rt}.
\end{eqnarray*}

Together with Eq. (\ref{Eq:x}), by the total probability formula, the lemma is proved.

\section{Proof of Lemma 3}
\label{proof:l3}

Under the LCFS discipline, if a status packet arrives when the battery is not empty or afterwards a energy packet arrives before another status packet enters the data buffer, the status packet is a valid update. Thus, by the PASTA property, the probability of a status packet being a valid update is
\begin{eqnarray*}
\sum_{i=-B}^{-1}\mathrm{Prob}\left\{S_n=i\right\} + \frac{1}{1+\theta}\sum_{i=0}^{K}\mathrm{Prob}\left\{S_n=i\right\} = \frac{(\theta^{-B} - 1)\left(1+\theta\right)+(1-\theta^{K+1})}{\left(\theta^{-B}-\theta^{K+1}\right)\left(1+\theta\right)},
\end{eqnarray*}
and the arrival rate of valid update is
$\tilde{\lambda} = \lambda\frac{(\theta^{-B} - 1)\left(1+\theta\right)+(1-\theta^{K+1})}{\left(\theta^{-B}-\theta^{K+1}\right)\left(1+\theta\right)}. $

\section{Proof of Lemma 4}
\label{proof:l4}
We first obtain the probability distribution of the state $S_n^-$ before valid updates' arrivals. By the definition of conditional probability and the PASTA property, we have
\begin{eqnarray}
\label{Eq:i}
\mathrm{Prob}\left\{S_n^-=i\right\}&=&\frac{\mathrm{Prob}\left\{S_n=i \mathrm{~upon~arrival}, \mathrm{valid~update}\right\}}{\Prob{\mathrm{valid~update}}}\notag\\
&=&\left\{
\begin{aligned}
&\frac{\left(1-\theta\right)\theta^{i}}{\left(\theta^{-B}-1\right)\left(1+\theta\right)+\left(1-\theta^{K+1}\right)},&\mathrm{~if~}0\leq i\leq K;\\
&\frac{\left(1-\theta^2\right)\theta^{i}}{\left(\theta^{-B}-1\right)\left(1+\theta\right)+\left(1-\theta^{K+1}\right)},&\mathrm{~if~}-B\leq i\leq-1.
\end{aligned}
\right.
\end{eqnarray}

Since the status packets arriving when the battery is not empty have zero sojourn time, we have
\begin{eqnarray}
\label{Eq:t1}
\mathrm{Prob}\left\{T_n \leq t\right\} = \sum_{i=-B}^{-1}\mathrm{Prob}\left\{S_n^-=i\right\} + \sum_{i=0}^{K}\mathrm{Prob}\left\{S_n^-=i\right\}\mathrm{Prob}\left\{T_n \leq t|S_n^-=i\right\},
\end{eqnarray}
where $\mathrm{Prob}\left\{T_n \leq t|S_n^-=i\geq0\right\}$ is the conditional probability distribution of sojourn time given that an energy packet comes before another status packet after the $n$-th valid update, which leads to
\begin{equation}
\label{Eq:ti1}
\mathrm{Prob}\left\{T_n \leq t|S_n^-=i\right\} = (1+\theta)\int_0^tre^{-r\tau}e^{-\lambda\tau}\mathrm{d}\tau = 1-e^{-\left(\lambda+r\right)t}, \mathrm{~if~}i \geq 0.
\end{equation}
Substituting Eq. (\ref{Eq:i}) and (\ref{Eq:ti1}) into Eq. (\ref{Eq:t1}), we have
\begin{eqnarray}
\label{Eq:tl}
\mathrm{Prob}\left[T_n \leq t\right]
&=&1-e^{-\left(\lambda+r\right)t}\frac{\left(1-\theta^{K+1}\right)}{\left(\theta^{-B}-1\right)\left(1+\theta\right)+\left(1-\theta^{K+1}\right)}
\end{eqnarray}

For the CDF of peak age $A$, we first obtain the conditional probability of the inter-delivery time $D_n$ between the delivery of the $n$-th and the $(n+1)$-th valid update given the sojourn time $T_n$ of the $n$-th valid update. After that, the CDF of peak age $A$ can be given by the total probability formula:
\begin{equation}
\label{Eq:al}
\Prob{A\leq a} = \int_0^af_{T_n}(a-d)\mathrm{P}\left\{D_n\leq d|T_n=a-d\right\}\,\mathrm{d}d
\end{equation}

Denote the system state right after the delivery of the $n$-th valid update as $S_n^+$. By the total probability formula, we have
\begin{eqnarray*}
\mathrm{Prob}\left\{D_n \leq d|T_n=t\right\} = \sum_{i=-B+1}^{K}\mathrm{Prob}\left\{D_n\leq d|T_n=t, S_n^+ = i\right\}\mathrm{Prob}\left\{ S_n^+ = i|T_n=t\right\}.
\end{eqnarray*}
Due to the fact that the inter-delivery time $D_n$ between the delivery of the $n$-th and the $(n+1)$-th valid update and the sojourn time $T_n$ of the $n$-th valid update are conditionally independent given the state $S_n^+$ right after the $n$-th valid update's delivery, the former equation can be simplified as
\begin{eqnarray}
\label{Eq:dt1}
\mathrm{Prob}\left\{D_n \leq d|T_n=t\right\} = \sum_{i=-B+1}^{K}\mathrm{Prob}\left\{D_n\leq d|S_n^+ = i\right\}\mathrm{Prob}\left\{ S_n^+ = i|T_n=t\right\}.
\end{eqnarray}

First, we obtain the first term on the right-hand side of Eq. (\ref{Eq:dt1}). If $S_n^+ < 0$, the inter-delivery time equals the inter-arrival time of two successive status packets, which gives
\begin{equation}
\label{Eq:di1}
\mathrm{Prob}\left\{D_n\leq d|S_n^+ = i\right\} = 1-e^{\lambda d}, \mathrm{~if~}i<0.
\end{equation}
If $S_n^+ \geq 0$, the complement of $D_n\leq d$ is that in $d$ after the $n$-th valid update's delivery, either there is no status packet arrival or there is no energy arrival, or that no more than $S_n^+$ energy packets arrive and all the energy packets arriving in $d$ after the $n$-th valid update's delivery come before status packets. Therefore, if $S_n^+ = 0$, we have
\begin{eqnarray}
\label{Eq:di2}
\mathrm{Prob}\left\{D\leq d|S_n^+ = 0\right\}&=&\mathrm{Prob}\left\{\mathrm{both~data~arrivals~and~energy~arrivals~in~}d\right\}\notag\\
&=&1-e^{-\lambda d}-e^{-rd}+e^{-r\left(1+\theta\right)d}.
\end{eqnarray}
If $S_n^+ > 0$, we have
\begin{eqnarray}
\label{Eq:di3}
&~&\mathrm{Prob}\left\{D_n\leq d|S_n^+ = i\right\}\notag\\
&=&\mathrm{Prob}\left\{\mathrm{both~data~arrivals~and~energy~arrivals~in~}d\right\} \notag\\
&~&- \sum_{k=1}^{i}\mathrm{Prob}\left\{k\mathrm{~energy~arrivals~in~}d, \mathrm{~all~energy~arrivals~before~the~first~data~arrival}\right\}\notag\\
&=& 1-e^{-\lambda d}-\frac{\theta^{-i}-\theta}{1-\theta}e^{-rd}-e^{-r\left(1+\theta\right)d}\left[\frac{\theta}{1-\theta}\sum_{k=0}^{i}\frac{\left(rd\right)^k}{k!}-\frac{\theta^{-i}}{1-\theta}\sum_{k=0}^{i}\frac{\left(r\theta d\right)^k}{k!}\right].
\end{eqnarray}

Next, we obtain the second term on the right-hand side of Eq. (\ref{Eq:dt1}). If $T_n = 0$, we have that the state $S_n^-$ right before the $n$-th valid update's arrival is negative, and that $S_n^+ = S_n^- + 1$. Therefore, we get
\begin{eqnarray}
\label{Eq:it1}
\mathrm{Prob}\left\{ S_n^+ = i|T_n=0\right\} &=& \mathrm{Prob}\left\{ S_n^- = i-1|S_n^- < 0\right\}
=\left\{
\begin{aligned}
&0,&\mathrm{~if~} i\geq 1;\\
&\frac{\left(1-\theta\right)\theta^{i+B-1}}{1-\theta^{B}},&\mathrm{~if~} i\leq 0.
\end{aligned}
\right.
\end{eqnarray}
The sufficient and necessary condition of $T_n > 0$ is $S_n^- \geq 0$. If $T_n > 0$, the valid update first enters the data buffer, and is deliveried when an energy packet arrives. Therefore, $S_n^- = S_n^+$, and
\begin{eqnarray}
\label{Eq:it2}
\mathrm{Prob}\left\{ S_n^+ = i|T_n=t\right\} &=& \mathrm{Prob}\left\{ S_n^- = i|S_n^- \geq 0\right\}
=\left\{
\begin{aligned}
&0,&\mathrm{~if~} i<0;\\
&\frac{\left(1-\theta\right)\theta^{i}}{1-\theta^{K+1}},&\mathrm{~if~} i\geq 0.
\end{aligned}
\right.
\end{eqnarray}

Substituting Eq. (\ref{Eq:di1})--(\ref{Eq:it2}) into Eq. (\ref{Eq:dt1}), we get
\begin{eqnarray}
\label{Eq:dt}
&~&\mathrm{Prob}\left[D_n \leq d|T_n=t\right] \notag\\
& = & \left\{
\begin{aligned}
&\left(1-e^{-\lambda d}\right)\left[1-e^{-rd}\frac{\theta^{B-1}-\theta^B}{1-\theta^B}\right],&t=0;\\
&1-e^{-\lambda d} - \frac{K - \frac{\theta^2 - \theta^{K+2}}{1-\theta} }{1-\theta^{K+1}}e^{-rd}&~\\
&~~-\frac{e^{-r\left(1+\theta\right)d}}{1-\theta^{K+1}}\left[ - \sum_{k=0}^{K-1}(K-k)\frac{\left(r\theta d\right)^k}{k!} + \frac{\theta^2}{1-\theta}\sum_{k=0}^{K-1}\frac{\left(r\theta d\right)^k}{k!} - \frac{\theta^{K+2}}{1-\theta}\sum_{k=0}^{K-1}\frac{\left(rd\right)^{k}}{k!}\right] &t >0.
\end{aligned}
\right.
\end{eqnarray}

Combining Eq. (\ref{Eq:tl}) and (\ref{Eq:dt}), the lemma is proved.


\end{appendices}

\bibliographystyle{ieeetr}

\end{document}